\documentclass[times,11pt,journal,onecolumn]{IEEEtran}

\usepackage{cite}
\usepackage{fixltx2e}
\usepackage[cmex10]{amsmath}
\usepackage{array}
\usepackage{wasysym}
\usepackage{dsfont}
\usepackage[latin1]{inputenc}                       
\usepackage[english]{babel}                         
\usepackage[T1]{fontenc}
\usepackage{mathtools}
\usepackage{amssymb} 
\usepackage{enumerate}
\usepackage{bbm}
\usepackage{epsfig,syntonly}
\usepackage{verbatim,times}
\usepackage{epstopdf}
\usepackage{graphicx}
\usepackage{xcolor, accents}
\usepackage{latexsym,fancyhdr,bm}
\usepackage{url}

\DeclareMathSymbol{\widehatsym}{\mathord}{largesymbols}{"62}

\newtheorem{theorem}{Theorem}

\newtheorem{rmk}[theorem]{Remark}
\newtheorem{cor}[theorem]{Corollary}

\newtheorem{lemma}[theorem]{Lemma}

\makeatletter
\let\l@ENGLISH\l@english
\makeatother

\begin{document}

\title{Unified Scaling of Polar Codes: Error Exponent, Scaling Exponent, Moderate Deviations, and Error Floors}
\author{Marco~Mondelli, S.~Hamed~Hassani, and~R\"{u}diger~Urbanke%
\thanks{M. Mondelli and R. Urbanke are with the School of Computer and Communication Sciences,
EPFL, CH-1015 Lausanne, Switzerland
(e-mail: \{marco.mondelli, ruediger.urbanke\}@epfl.ch).

S. H. Hassani is with the Computer Science Department, ETH Z\"{u}rich, Switzerland
(e-mail: hamed@inf.ethz.ch).
}
}

\maketitle

\begin{abstract}
Consider the transmission of a polar code of block length $N$ and rate
$R$ over a binary memoryless symmetric channel $W$ and let $P_{\rm
e}$ be the block error probability under successive cancellation decoding.
In this paper, we develop new bounds that characterize the relationship
of the parameters $R$, $N$, $P_{\rm e}$, and the quality of the
channel $W$ quantified by its capacity $I(W)$ and its Bhattacharyya
parameter $Z(W)$.

In previous work, two main regimes were studied. In the
\emph{error exponent} regime, the channel $W$ and the rate $R<I(W)$
are fixed, and it was proved that the error probability $P_{\rm
e}$ scales roughly as $2^{-\sqrt{N}}$. In the \emph{scaling exponent}
approach, the channel $W$ and the error probability $P_{\rm e}$ are
fixed and it was proved that the gap to capacity $I(W)-R$
scales as $N^{-1/\mu}$. Here, $\mu$ is called \emph{scaling exponent}
and this scaling exponent depends on the channel $W$. A heuristic computation for the
binary erasure channel ($\rm BEC$) gives $\mu=3.627$ and it was shown that, for any channel $W$, $3.579 \le \mu \le 5.702$.

Our contributions are as follows. First, we
provide the tighter upper bound $\mu \le 4.714$ valid for any $W$.
With the same technique, we obtain the upper bound $\mu \le 3.639$ for the case of the $\rm BEC$; this upper bound approaches very closely the
heuristically derived value for the scaling exponent of the erasure channel.

Second, we develop a trade-off between the gap to capacity $I(W)-R$
and the error probability $P_{\rm e}$ as functions of the block
length $N$. In other words, we neither fix the gap to capacity
(error exponent regime) nor the error probability (scaling exponent
regime), but we do consider a \emph{moderate deviations} regime in
which we study how fast both quantities, as functions of the block length $N$, simultaneously go to $0$.

Third, we prove that polar codes are not affected by \emph{error
floors}. To do so, we fix a polar code of block length $N$ and rate
$R$.  Then, we vary the channel $W$ and study the impact of this variation on the error probability.  We show that the error
probability $P_{\rm e}$ scales as the Bhattacharyya parameter $Z(W)$
raised to a power that scales roughly like ${\sqrt{N}}$. This agrees with
the scaling in the error exponent regime.  \end{abstract}


\section{Introduction} \label{sec:intro}

{\bf Performance Analysis in Different Regimes.} When we consider
the transmission over a channel $W$ by using a coding scheme, the
parameters of interest are the rate $R$, that represents the amount
of information transmitted per channel use, the block length $N$,
that represents the total number of channel uses, and the block
error probability $P_{\rm e}$. The exact characterization of the
relationship of $R$, $N$, $P_{\rm e}$, and the quality of the
channel $W$ (which can be quantified, e.g., by its capacity $I(W)$
or its Bhattacharyya parameter $Z(W)$) is a formidable task. It is
easier to study the {\em scaling} of these parameters in various
regimes, i.e., by fixing some of these parameters and by considering
the relationship among the remaining parameters.

Concretely, consider the plots in Figure \ref{fig:scalingreg}:
they represent the performance of a family of codes $\mathcal
C$ with rate $R=0.5$. Different curves correspond to codes of different
block lengths $N$. The codes are transmitted over a family of channels
$\mathcal{W}$ parameterized by $z$, that is represented on the horizontal axis. On the vertical axis, we represent the error probability $P_{\rm e}$. The error probability is an increasing function of $z$, which means that the channel gets ``better'' as $z$ decreases. The parameter $z$ indicates the quality of the transmission channel $W$ and, for example, it could be set to $Z(W)$ or to $1-I(W)$. Let us assume that there exists a threshold $z^*$ such that,
if $z< z^*$, then $P_{\rm e}$ tends to $0$ as $N$ grows large, whereas
if $z> z^*$, then $P_{\rm e}$ tends to $1$ as $N$ grows large. For example, if the family of codes $\mathcal C$ is capacity achieving, then we
can think of the threshold $z^*$ as the channel parameter such that
$I(W)=R$. In the example of Figure \ref{fig:scalingreg}, we have
that $z^*=0.5$.

The oldest approach for analyzing the performance of such a family $\mathcal C$ is known as \emph{error exponent}. We pick any channel parameter $z< z^*$. Then, by definition of $z^*$, the error probability tends to $0$ as $N$ grows large. The error exponent quantifies this statement and computes how the error probability varies as a function of the block length. This approach is pictorially represented as the vertical/blue cut in Figure \ref{fig:scalingreg}. The best possible scaling is obtained by considering random codes, that give
\begin{equation*}
P_{\rm e} = e^{-N E(R, W)+o(N)},
\end{equation*}
where $E(R, W)$ is the so-called error exponent \cite{Gal65}. 

Another approach is known as \emph{scaling exponent}. We pick a target error probability $P_{\rm e}$. Then, by definition of $z^*$, the gap between the threshold and the channel parameter $z^*-z$ tends to $0$ as $N$ grows large. The scaling exponent quantifies this statement and computes how the gap to the threshold varies as a function of the block length. This approach is pictorially represented as the horizontal/red cut in Figure \ref{fig:scalingreg}. From a practical viewpoint, we are interested in such a regime, as we typically have a certain requirement on the error probability and look for the shortest code possible for transmitting over the assigned channel. For specific classes of codes, this approach was put forward in \cite{Tillich-cpc00, Mon01b}. As a benchmark, a sequence of works starting from \cite{Do61}, then \cite{St62}, and finally \cite{Ha09, PPV10} shows that the smallest possible block length $N$ required to achieve a gap $z^*-z$ to the threshold with a fixed error probability $P_{\rm e}$ is such that
\begin{equation}\label{eq:randomsc}
N \approx \frac{V(Q^{-1}(P_{\rm e}))^2}{(z^*-z)^2},
\end{equation}
where $Q(\cdot)$ is the tail probability of the standard normal distribution; and $V$ is referred to as channel dispersion and measures the stochastic variability of the channel relative to a deterministic channel with the same capacity. In general, if $N$ is $\Theta(1/(z^*-z)^\mu)$, then we say that the family of codes $\mathcal C$ has scaling exponent $\mu$. Hence, by \eqref{eq:randomsc}, the most favorable scaling exponent is $\mu = 2$ and is achieved by random codes. Furthermore, for a large class of ensembles of LDPC codes and channel models, the scaling exponent is also $\mu = 2$. However, it has to be pointed out that the threshold of such LDPC ensembles does not converge to capacity \cite{AMRU09}.

In summary, in the error exponent regime, we compute how fast $P_{\rm e}$ goes to $0$ as a function of $N$ when $z^*-z$ is fixed; and in the scaling exponent regime, we compute how fast $z^*-z$ goes to $0$ as a function of $N$ when $P_{\rm e}$ is fixed. Then, a natural question is to ask how fast do \emph{both} $P_{\rm e}$ \emph{and} $z^*-z$ go to $0$ as functions of $N$. In other words, we can describe a trade-off between the speed of decay of the error probability and the speed of decay of the gap to capacity as functions of the block length. This intermediate approach is named the \emph{moderate deviations} regime and is studied for random codes in \cite{AW14}.

The last scaling approach we consider concerns the so-called
\emph{error floor} regime. We pick a code of assigned block length $N$ and rate $R$. Then, we compute how the error probability $P_{\rm e}$ behaves as a function of the channel parameter $z$. This corresponds to taking
into account one of the four curves in Figure \ref{fig:scalingreg}. This is a notion that became important when iterative coding schemes were introduced. For such 
schemes, it was observed that frequently the individual curves $P_{\rm
e}(z)$ show an abrupt change of slope, from very steep to very shallow, when going from bad channels to good channels (see, e.g., Figure \ref{fig:floors}). The region where the slope is very
shallow was dubbed the error floor region. More specifically, if we consider a parallel concatenated turbo code, then there is a fixed number of low-weight codewords, regardless of the block length $N$ (see Section 6.9 of \cite{RiU08}). The same behavior can be observed for the ensemble average of LDPC codes, when the minimal variable-node degree is equal to $2$. This means that, in the error floor region, the block error probability is dominated by a term that is independent of $N$ and scales as $z^w$, where $w$ denotes the minimal weight of a non-zero codeword. If the minimal variable-node degree is at least $3$, then the number of low-weight codewords vanishes with $N$ and the block error probability scales as $z^w/N^{w(\mathtt{l}_{\rm min}/2-1)}$. For a more precise statement, see Theorem D.32 in Appendix D of \cite{RiU08}. In this paper, we will show that polar codes have a much more favorable behavior, i.e., the block error probability scales roughly as $z^{\sqrt{N}}$.

\begin{figure}[t] 
\centering 
\includegraphics[width=0.8\columnwidth]{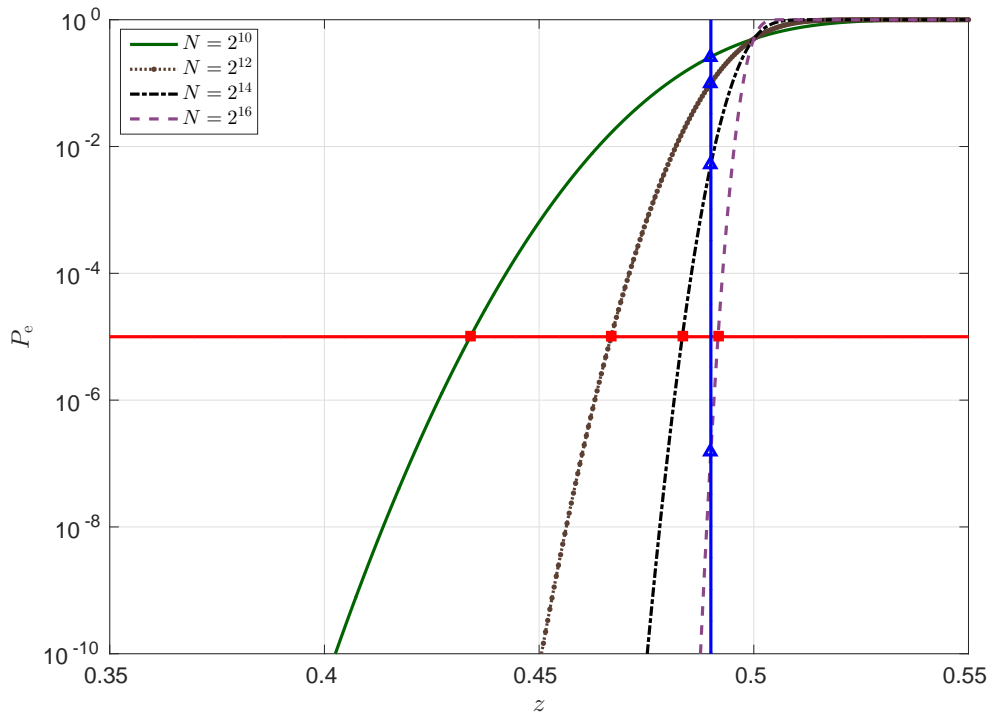}
\caption{Performance of the family of codes $\mathcal C$ with rate $R=0.5$ transmitted over the family of channels $\mathcal{W}$ with threshold $z^*=0.5$. Each curve corresponds to a code of an assigned block length $N$; on the $x$-axis it is represented the channel parameter $z$; and on the $y$-axis the error probability $P_{\rm e}$. The error exponent regime captures the behavior of the blue vertical cuts of fixed channel parameter $z$ (or, equivalently, of fixed gap to threshold $z^*-z$). The scaling exponent regime captures the behavior of the red horizontal cuts of fixed error probability $P_{\rm e}$. The error floor regime captures the behavior of a single curve of fixed block length $N$.} 
\label{fig:scalingreg}
\end{figure} 

\begin{figure}[tb] 
\centering 
\includegraphics[width=0.8\columnwidth]{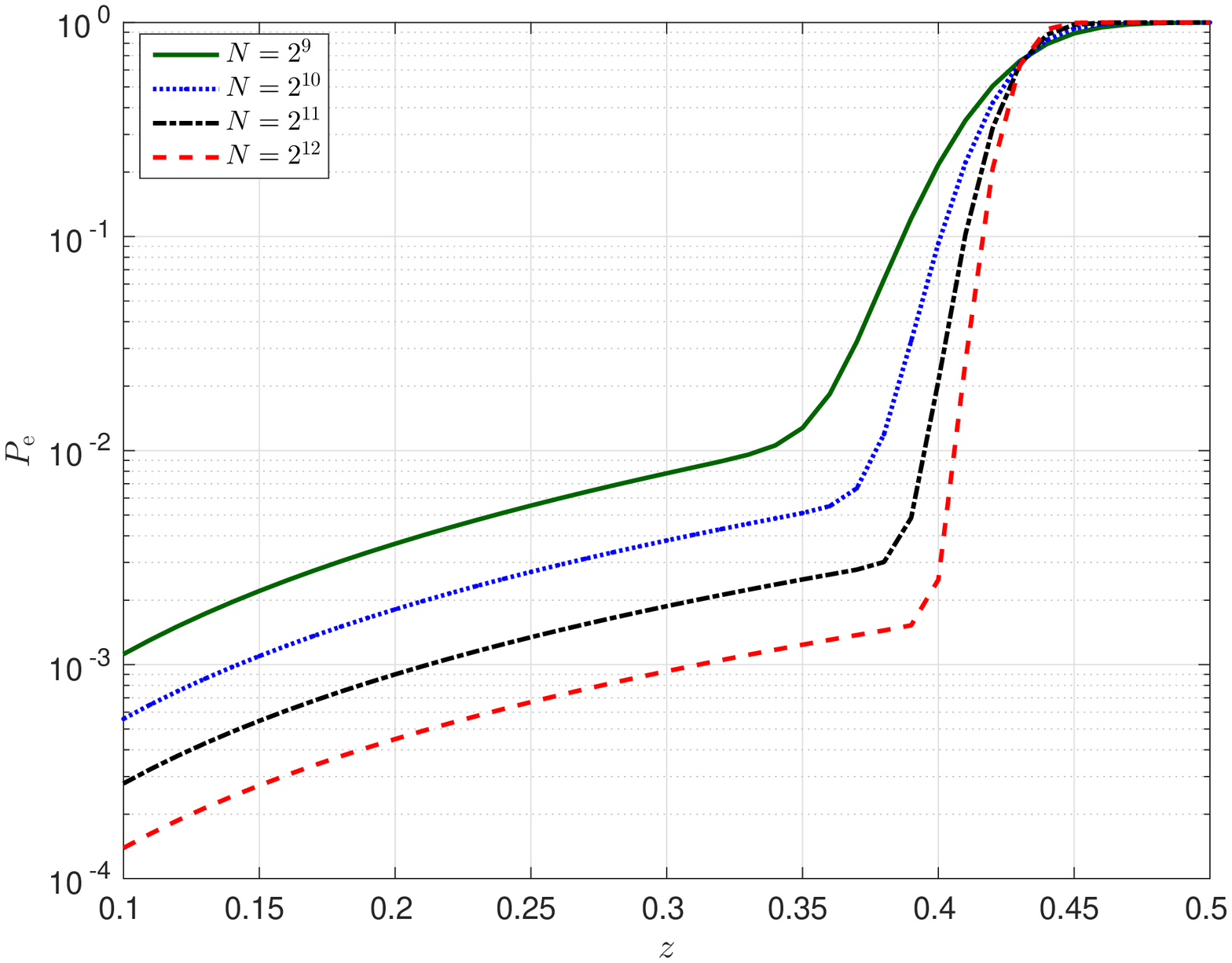}
\caption{Performance of the family of $(3, 6)$-regular LDPC codes transmitted over the binary erasure channel with erasure probability $z$. The waterfall region in which the error probability decreases sharply is clearly distinguishable from the error floor region in which the decay is much slower.} 
\label{fig:floors}
\end{figure} 

{\bf Existing Results for Polar Codes.} Polar codes have attracted the interest of the scientific community, as they provably achieve the capacity of a large class of channels, including any binary memoryless symmetric channel ($\rm BMSC$), with low encoding and decoding complexity. Since their introduction in the seminal paper \cite{Ari09}, the performance of polar codes has been extensively studied in different regimes. 

Concerning the \emph{error exponent} regime, in \cite{ArT09} it is proved that the block error probability under successive cancellation (SC) decoding behaves roughly as $2^{-\sqrt{N}}$. This result is further refined in \cite{HMTU13}, where it is shown that
$\log_2(-\log_2 P_{\rm e})$ scales as $$\displaystyle\frac{\log_2 N}{2}+\displaystyle\frac{\sqrt{\log_2 N}}{2}\cdot
Q^{-1}\left(\displaystyle\frac{R}{C}\right) + o(\sqrt{\log_2 N}).$$ This last result holds both under SC decoding and under optimal MAP decoding. 

Concerning the \emph{scaling exponent}\footnote{In \cite{KMTU10}, the scaling exponent is defined as the value of $\mu$ such that $$\lim_{N\rightarrow \infty, N^{1/\mu}(C-R)=z}P_{\rm e}(N, R, C)=f(z),$$ for some function $f(z)$. However, it is an open question to prove that such a limit exists.} regime, the value of $\mu$ depends on the particular channel taken into account. The authors of \cite{KMTU10} provide a heuristic method for computing the scaling exponent for transmission over the $\rm BEC$ under SC decoding; this method yields $\mu \approx 3.627$. Furthermore, in \cite{XG13} it is shown that the block length scales polynomially fast with the inverse of the gap to capacity, while the error probability is upper bounded by $2^{-N^{0.49}}$. Universal bounds on $\mu$, valid for any $\rm BMSC$ under SC decoding, are presented in \cite{HAU14}: the scaling exponent is lower bounded by $3.579$ and is upper bounded by $6$. In addition, it is conjectured that the lower bound on $\mu$ can be increased up to $3.627$, i.e., up to the value heuristically computed for the $\rm BEC$. The upper bound on $\mu$ is further refined to $5.702$ in \cite{GB14}. As a significant performance gain was obtained by using a successive cancellation list (SCL) decoder \cite{TVa15}, the scaling exponent of list decoders was also studied. However, in \cite{MHU14list-ieeeit} it is proved that the value of $\mu$ does not change by adding a list of any finite size to the MAP decoder. In addition, when transmission takes place over the $\rm BEC$, the scaling exponent stays the same also under genie-aided SC decoding for any finite number of helps from the genie. 

Concerning the \emph{error floor} regime, in \cite{EP13} it is proved that the stopping distance of polar codes scales as $\sqrt{N}$, which implies good error floor performance under belief propagation (BP) decoding. The authors of \cite{EP13} also provide simulation results that show no sign of error floor for transmission over the $\rm BEC$ and over the binary additive white Gaussian noise channel (BAWGNC).

{\bf Contribution of the Present Work.} In this paper, we provide a unified view on the performance analysis of polar codes and present several results about the scaling of the parameters of interest, namely, the rate $R$, the block length $N$, the error probability under SC decoding $P_{\rm e}$, and the quality of the channel $W$. In particular, our contributions address the \emph{scaling exponent}, the \emph{moderate deviations}, and the \emph{error floor} regimes, and we summarize them as follows.

\begin{enumerate}

\item \emph{New universal upper bound on the scaling exponent $\mu$}. We show that $\mu\le 4.714$ for any $\rm BMSC$ and that $\mu \le 3.639$ for the $\rm BEC$. Basically, this result improves by $1$ the previous upper bound valid for any $\rm BMSC$ and approaches closely the value $3.627$ that has been heuristically computed for the $\rm BEC$. The proof technique consists in relating the scaling exponent to the supremum of some function and, then, in describing an interpolation algorithm to obtain a provable upper bound on this supremum. The values $4.714$ for any $\rm BMSC$ and $3.639$ for the $\rm BEC$ are obtained for a particular number of samples used by the algorithm and they can be slightly improved simply by running the algorithm with a larger number of samples.

\item \emph{Moderate deviations: joint scaling of error probability and gap to capacity}. We
unify the two perspectives of the error exponent and the scaling exponent by letting both the gap to capacity $I(W)-R$ and the error probability $P_{\rm e}$ to go to $0$ as functions of the block length $N$. In particular, we describe a trade-off between the speed of decay of $P_{\rm e}$ and the speed of decay of $I(W)-R$. In the limit in which the gap to capacity is arbitrarily small but independent of $N$, this trade-off recovers the result of \cite{ArT09}, where it is shown that $P_{\rm e}$ scales roughly as $2^{-\sqrt{N}}$. 

\item \emph{Absence of error floors}. We prove that polar codes are not affected by error floors. To do so, we consider a polar code of block length $N$ and rate $R$ designed for transmission over a channel $W'$. Then, we look at the performance of this fixed code over other channels $W$ that are ``better'' than $W'$; and we study the error probability $P_{\rm e}$ as a function of the Bhattacharyya parameter $Z(W)$. Note that the code is fixed and the channel varies, which means that we do not choose the optimal polar indices for $W$. In particular, we prove that $P_{\rm e}$ scales roughly as $Z(W)^{\sqrt{N}}$, and this result is in agreement with the error exponent regime.
\end{enumerate}

The rest of the paper is organized as follows. In Section \ref{sec:prel}, we review some preliminary notions about polar coding. In the successive three sections, we describe our original contributions: in Section \ref{sec:scalexp}, we present the new upper bound on the scaling exponent; in Section \ref{sec:joint}, we address the moderate deviations regime; and in Section \ref{sec:floors}, we prove that polar codes are not affected by error floors. In Section \ref{sec:concl}, we conclude the paper with some final remarks.


\section{Preliminaries}\label{sec:prel}

Let $W$ be a $\rm BMSC$, and let $\mathcal{X}=\{0,1\}$ denote its input alphabet, 
$\mathcal{Y}$ the output alphabet, and $\{W(y \mid x) : x\in \mathcal{X}, y\in \mathcal{Y}\}$ the transition probabilities. Denote by $I(W) \in [0,1]$ the mutual information between the input and output of $W$ with uniform distribution on the input. Then, $I(W)$ is also equal to the capacity of $W$. Denote by $Z(W)\in [0,1]$ the Bhattacharyya parameter of $W$, which
is defined as
\begin{align*}
& Z(W)= \sum_{y \in \mathcal{Y}} \sqrt{W(y\mid 0)W(y \mid 1)},
\end{align*}
and it is related to the capacity $I(W)$ via
\begin{align}
&Z(W)+I(W)\ge1,\label{eq:capBatta0}\\ 
&Z(W)^2+I(W)^2\le1, \label{eq:capBatta}
\end{align}
both proved in~\cite{Ari09}. 

The basis of channel polarization consists in mapping two identical copies of the channel $W: \mathcal{X}\to \mathcal{Y}$ into the pair of channels $W^0: \mathcal{X}\to \mathcal{Y}^2$ and $W^1:\mathcal{X}\to \mathcal{X}\times\mathcal{Y}^2$, defined as \cite[Section I-B]{Ari09}, \cite[Section I-B]{HAU14},
\begin{equation}
\begin{split}
W^0(y_1, y_2\mid x_1) & = \sum_{x_2\in \mathcal X} \frac{1}{2}W(y_1\mid x_1 \oplus x_2) W(y_2\mid x_2),\\
W^1(y_1, y_2, x_1\mid x_2) & = \frac{1}{2}W(y_1\mid x_1 \oplus x_2) W(y_2\mid x_2).\\
\end{split}
\end{equation}
Then, the idea is that $W^0$ is a ``worse'' channel and $W^1$ is a ``better'' channel than $W$. This statement can be quantified by computing the relations among the Bhattacharyya parameters of $W$, $W^0$ and $W^1$:
\begin{align}
Z(W)\sqrt{2-Z(W)^2}&\le Z(W^0)\le 2Z(W)-Z(W)^2,\label{eq:minusB}\\
&Z(W^1)=Z(W)^2,\label{eq:plusB}
\end{align}
which follow from Proposition 5 of \cite{Ari09} and from Exercise 4.62 of \cite{RiU08}. In addition, when $W$ is a $\rm BEC$, we have that $W^0$ and $W^1$ are also $\rm BEC$s and, by Proposition 5 of \cite{Ari09}, 
\begin{equation}\label{eq:minusBEC}
Z(W^0)= 2Z(W)-Z(W)^2.
\end{equation}
By repeating this operation $n$ times, we map $2^n$ identical copies of $W$ into the synthetic channels $W_n^{(i)}$ ($i\in \{1, \cdots, 2^n\}$), defined as 
\begin{equation}\label{eq:defWni}
W_n^{(i)} = (((W^{b_1^{(i)}})^{b_2^{(i)}})^{\cdots})^{b_n^{(i)}},
\end{equation}
where $(b_1^{(i)}, \cdots, b_n^{(i)})$ is the binary representation of the integer $i-1$ over $n$ bits. 

Given a $\rm BMSC$ $W$, for $n\in \mathbb N$, define a random sequence of channels $W_n$, as $W_0=W$, and 
\begin{equation}
W_{n} = \left\{ \begin{array}{ll}W_{n-1}^0, & \mbox{ w.p. } 1/2,\\ W_{n-1}^1,& \mbox{ w.p. } 1/2.\\ \end{array}\right.
\end{equation}
Let $Z_n(W)=Z(W_n)$ be the random process that tracks the Bhattacharyya parameter of $W_n$. Then, from \eqref{eq:minusB} and \eqref{eq:plusB} we deduce that, for $n\ge 1$,
\begin{equation}\label{eq:eqBMSC}
Z_{n} \left\{ \begin{array}{ll}\in \left[Z_{n-1}\sqrt{2-Z^2_{n-1}},\, 2 Z_{n-1}-Z^2_{n-1}\right], & \mbox{ w.p. } 1/2,\\ =Z^2_{n-1},& \mbox{ w.p. } 1/2.\\ \end{array}\right.
\end{equation}
When $W$ is a $\rm BEC$ with erasure probability $z$, then the process $Z_n$ has a simple closed form. It starts with $Z_0 = z$, and, by using \eqref{eq:minusB} and \eqref{eq:plusB}, we deduce that, for $n\ge 1$, 
\begin{equation}\label{eq:recBEC}
Z_{n} = \left\{ \begin{array}{ll}2 Z_{n-1}-Z^2_{n-1}, & \mbox{ w.p. } 1/2,\\ Z^2_{n-1},& \mbox{ w.p. } 1/2.\\ \end{array}\right.
\end{equation}

Consider the transmission over $W$ of a polar code of block length $N=2^n$ and rate $R$ and let $P_{\rm e}$ denote the block error probability under SC decoding. Then, by Proposition 2 of \cite{Ari09},
\begin{equation}\label{eq:ubPe}
P_{\rm e} \le \sum_{i\in \mathcal I} Z_n^{(i)},
\end{equation}
where $Z_n^{(i)}$ denotes the Bhattacharyya parameter of $W_n^{(i)}$ and $\mathcal I$ denotes the information set, i.e., the set containing the positions of the information bits.


\section{New Universal Upper Bound on the Scaling Exponent}\label{sec:scalexp}

In this section, we propose an improved upper bound on the scaling exponent that is valid for the transmission over any $\rm BMSC$ $W$.  First of all, we relate the value of the scaling exponent $\mu$ to the supremum of some function. Second, we provide a provable bound on this supremum, which gives us a provably valid choice for $\mu$, i.e., $\mu=4.714$ for any $\rm BMSC$ and $\mu=3.639$ for the $\rm BEC$. More specifically, in Section \ref{subsec:scalstat}, we present the statement and the discussion of these two main theorems. In Sections \ref{subsec:mathpart} and \ref{subsec:designh}, we give the proof of the first and of the second result, respectively.

\subsection{Main Result: Statement and Discussion}\label{subsec:scalstat}

\begin{theorem}[From Eigenfunction to Scaling Exponent]\label{th:scalingexp}
Assume that there exists a function $h(x): [0, 1] \to [0, 1]$ such that $h(0)=h(1)=0$, $h(x)>0$ for any $x\in (0, 1)$, and, for some $\mu > 2$,
\begin{equation}\label{eq:suph}
\displaystyle\sup_{\substack{x\in (0, 1), y \in [x\sqrt{2-x^2}, 2x-x^2]}}\displaystyle\frac{h(x^2)+h(y)}{2h(x)} < 2^{-1/\mu}.
\end{equation}
Consider the transmission over a $\rm BMSC$ $W$ with capacity $I(W)$ by using a polar code of rate $R<I(W)$. Fix $p_{\rm e} \in (0, 1)$ and assume that the block error probability under successive cancellation decoding is \emph{at most} $p_{\rm e}$. Then, it suffices to have a block length $N$ such that 
\begin{equation}\label{eq:scalingexp}
\begin{split}
N &\le \frac{\beta_1}{(I(W)-R)^{\mu}},
\end{split}
\end{equation}
where $\beta_1$ is a universal constant that does not depend on $W$, but only on $p_{\rm e}$. If $W$ is a $\rm BEC$, a less stringent hypothesis on $\mu$ is required for \eqref{eq:scalingexp} to hold. In particular, the condition \eqref{eq:suph} is replaced by 
\begin{equation}\label{eq:suphBEC}
\displaystyle\sup_{x\in (0, 1)}\displaystyle\frac{h(x^2)+h(2x-x^2)}{2h(x)} < 2^{-1/\mu}.
\end{equation}
\end{theorem}

\begin{theorem}[Valid Choice for Scaling Exponent]\label{th:validchoice}
Consider the transmission over a $\rm BMSC$ $W$ with capacity $I(W)$ by using a polar code of rate $R<I(W)$. Fix $p_{\rm e} \in (0, 1)$ and assume that the block error probability under successive cancellation decoding is \emph{at most} $p_{\rm e}$. Then, it suffices to have a block length $N$ upper bounded by \eqref{eq:scalingexp} with $\mu=4.714$ Furthermore, if $W$ is a $\rm BEC$, then \eqref{eq:scalingexp} holds with $\mu = 3.639$.
\end{theorem}

Before proceeding with the proofs, it is useful to discuss two points. The first remark focuses on the role of the function $h(x)$ and heuristically explains why the value of the scaling exponent is linked to the existence of a function that fulfills condition \eqref{eq:suph} (condition \eqref{eq:suphBEC} for the $\rm BEC$). The second remark points out that we can let the error probability to tend to 0 polynomially fast in $N$ and maintain the same scaling between gap to capacity and block length.

\begin{rmk}[Heuristic Interpretation of Function $h(x)$]
First, let $W$ be a $\rm BEC$ and consider the linear operator $T_{\rm BEC}$ defined as
\begin{equation}\label{eq:polarop}
T_{\rm BEC}(g) = \frac{g(z^2)+g(2z-z^2)}{2},
\end{equation}
where $g(z)$ is a bounded and real valued function over $[0, 1]$. The relation between the Bhattacharyya process $Z_n$ and the operator $T_{\rm BEC}$ is given by
\begin{equation}
{\mathbb E}\left[g(Z_n)\mid Z_0 = z\right] =\overbrace{T_{\rm BEC} \circ T_{\rm BEC} \circ \cdots \circ T_{\rm BEC}(g)}^{n \mbox{ }{\rm times}} = T_{\rm BEC}^n (g),
\end{equation}
where the formula comes from a straightforward application of \eqref{eq:recBEC}. A detailed explanation of the dynamics of the functions $T_{\rm BEC}^n (g)$ is provided in Section III of \cite{HAU14}. In short, a simple check shows that $\lambda =1$ is an eigenvalue of the operator $T_{\rm BEC}$ with eigenfunctions $v_0(z)=1$ and $v_1(z)=z$. Let $\lambda^*$ be the largest eigenvalue of $T_{\rm BEC}$ other than $\lambda = 1$, and define $\mu^*$ as $\mu^* = -1/\log_2 \lambda^* $. Then, the heuristic discussion of \cite{HAU14} leads to the fact that $\mu^*$ is the largest candidate that we could plug in \eqref{eq:suphBEC}. For this choice, the function $h(x)$ represents the eigenfunction associated with the eigenvalue $\lambda^*$, namely,
\begin{equation}
\frac{h(x^2)+h(2x-x^2)}{2} = 2^{-1/\mu^*}h(x).
\end{equation}
A numerical method for the calculation of this second eigenvalue was originally proposed in \cite{KMTU10} and yields $\mu^* = 3.627$. Furthermore, in Section III of \cite{HAU14}, it is also heuristically explained how $\mu^*=3.627$ gives a lower bound to the scaling exponent of the $\rm BEC$.

Now, let $W$ be a $\rm BMSC$ and consider the operator $T_{\rm BMSC}$ defined as
\begin{equation}\label{eq:polaropgen}
T_{\rm BMSC}(g) = \displaystyle\sup_{y \in [x\sqrt{2-x^2}, 2x-x^2]}\frac{g(z^2)+ g(y)}{2}.
\end{equation}
Note that, differently from $T_{\rm BEC}$, the operator $T_{\rm BMSC}$ is not linear as it involves taking a supremum.
The relation between the Bhattacharyya process $Z_n$ and the operator $T_{\rm BMSC}$ is given by
\begin{equation} \label{T_BMSC}
{\mathbb E}\left[g(Z_n)\mid Z_0 = z\right] \le T_{\rm BMSC}^n (g),
\end{equation}
where the formula comes from a straightforward application of \eqref{eq:eqBMSC}. Similarly, as in the case of the $\rm BEC$, $\lambda =1$ is an eigenvalue of $T_{\rm BMSC}$, and we write the largest eigenvalue other than $\lambda =1$ as $2^{-1/\mu^*}$. Then, the idea is that $\mu^*$ is the largest candidate that we could plug in \eqref{eq:suph} and, for this choice, the function $h(x)$ represents the eigenfunction associated with the eigenvalue $2^{-1/\mu^*}$, namely, 
\begin{equation}\label{eq:eigenfungen}
\displaystyle\sup_{y \in [x\sqrt{2-x^2}, 2x-x^2]}\frac{h(x^2)+ h(y)}{2}=2^{-1/\mu^*}h(x).
\end{equation}
In Section IV of \cite{HAU14}, it is proved that the scaling exponent $\mu$ is upper bounded by $6$. This result is obtained by showing that the eigenvalue is at least $2^{-1/5}$, i.e. $\mu^*\le 5$, and that $\mu^*+1$ is an upper bound on the scaling exponent $\mu$. Furthermore, it is conjectured that $\mu^*$ is a tighter upper bound on the scaling exponent $\mu$. In \cite{GB14}, a more refined computation of $\mu^*$ is presented, which yields $\mu^*\le 4.702$, hence $\mu \le 5.702$. In this paper, we solve the conjecture of \cite{HAU14} by proving that, indeed, $\mu^*$ is an upper bound on the scaling exponent $\mu$. In addition, we show an algorithm that guarantees a \emph{provable} bound on the eigenvalue, thus obtaining $\mu\le 4.714$ for any $\rm BMSC$ and $\mu\le 3.639$ for the $\rm BEC$. We finally note from \eqref{T_BMSC} that  $T_{\rm BMSC}$ provides only an upper bound on the (expected) evolution of $Z_n$. As a result, although $\mu\le 4.714$ holds universally for any channel, this bound is certainly not tight if we consider a specific $\rm BMSC$. 
\end{rmk}

\begin{rmk}[Polynomial Decay of $P_{\rm e}$] \label{rmk:poly} With some more work, it is possible to prove the following generalization of Theorem \ref{th:scalingexp}. Assume that there exists $h(x)$ as in Theorem \ref{th:scalingexp} and consider the transmission over a $\rm BMSC$ $W$ with capacity $I(W)$ by using a polar code of rate $R<I(W)$. Then, for any $\nu>0$, the block length $N$ and the block error probability under successive cancellation decoding $P_{\rm e}$ are such that 
\begin{equation}\label{eq:jointPeN}
\begin{split}
P_{\rm e} &\le \frac{1}{N^\nu},\\
N &\le \frac{\beta_2}{(I(W)-R)^{\mu}},
\end{split}
\end{equation}
where $\beta_2$ is a universal constant that does not depend on the channel $W$. A sketch of the proof of this statement is given at the end of Section \ref{subsec:mathpart}. The result \eqref{eq:jointPeN} is a generalization of Theorem \ref{th:scalingexp} in the sense that, instead of being an assigned constant, the error probability goes to 0 polynomially fast in $1/N$, and the scaling between block length and gap to capacity, i.e., the value of $\mu$, stays the same. On the contrary, as described in Section \ref{sec:joint}, if the error probability is $O(2^{-N^\beta})$ for some $\beta\in (0, 1/2)$, then the scaling between block length and gap to capacity changes and depends on the exponent $\beta$.
\end{rmk}

\subsection{From Eigenfunction to Scaling Exponent: Proof of Theorem \ref{th:scalingexp}}\label{subsec:mathpart}

The proof of Theorem \ref{th:scalingexp} relies on the following two auxiliary results: Lemma \ref{lm:exptoBatta}, proved in Appendix \ref{app:exptoBatta}, relates the number of synthetic channels with a Bhattacharyya parameter small enough to an expected value over the Bhattacharyya process; and Lemma \ref{lm:htoexp}, proved in Appendix \ref{app:htoexp}, relates the expected value over the Bhattacharyya process to the function $h(x)$.

\begin{lemma}[From Expectation to Scaling Exponent] \label{lm:exptoBatta}
Let $Z_n(W)$ be the Bhattacharyya process associated with the channel $W$. Pick any $\alpha\in (0, 1)$ and assume that, for $n \ge 1$ and for some $\rho \le 1/2$,
\begin{equation}\label{eq:hpexBatta}
{\mathbb E}\left[(Z_n(1-Z_n))^{\alpha}\right] \le c_1 2^{-n\rho},
\end{equation}
where $c_1$ is a constant that does not depend on $n$. Then,
\begin{equation}\label{eq:gaplemma}
{\mathbb P}\left(Z_n \le p_{\rm e}\hspace{0.1em} 2^{-n}\right)\ge I(W)- c_2 \hspace{0.1em} 2^{-n(\rho-\alpha)},
\end{equation}
where $c_2= \sqrt{2 \hspace{0.1em}p_{\rm e}}+2\hspace{0.1em}c_1\hspace{0.1em}p_{\rm e}^{-\alpha}$. 
\end{lemma}

\begin{lemma}[From Eigenfunction to Expectation] \label{lm:htoexp}
Let $h(x): [0, 1] \to [0, 1]$ such that $h(0)=h(1)=0$, $h(x)>0$ for any $x\in (0, 1)$, and 
\begin{equation}\label{eq:suphrho}
\displaystyle\sup_{\substack{x\in (0, 1), y \in [x\sqrt{2-x^2}, 2x-x^2]}}\displaystyle\frac{h(x^2)+h(y)}{2h(x)} \le 2^{-\rho_1}.
\end{equation}
for some $\rho_1 \le 1/2$. Let $Z_n(W)$ be the Bhattacharyya process associated with the channel $W$. Pick any $\alpha \in (0, 1)$. Then, for any $\delta \in (0, 1)$, and for $n\in \mathbb N$,
\begin{equation}\label{eq:expZn}
{\mathbb E}\left[(Z_n(1-Z_n))^{\alpha}\right] \le \frac{1}{\delta}\left(2^{-\rho_1}+\sqrt{2}\frac{\delta}{1-\delta} c_3\right)^n,
\end{equation}
with $c_3$ defined as 
\begin{equation}\label{eq:defa}
c_3 = \sup_{x \in (\epsilon_1(\alpha), 1-\epsilon_2(\alpha))}\frac{(x(1-x))^\alpha}{h(x)},
\end{equation}
where $\epsilon_1(\alpha)$, $\epsilon_2(\alpha)$ denote the only two solutions in $[0, 1]$ of the equation
\begin{equation}\label{eps-solution}
\frac{1}{2}\left( \bigl(x(1+x)\bigr)^\alpha + \bigl((2-x)(1-x)^{1/3}\bigr)^\alpha \right)= 2^{-\rho_1}.
\end{equation} 
If $W$ is a $\rm BEC$, a less stringent hypothesis on $\rho_1$ is required for \eqref{eq:expZn} to hold. In particular, the condition \eqref{eq:suphrho} is replaced by
\begin{equation}\label{eq:suphrhoBEC}
\displaystyle\sup_{x\in (0, 1)}\displaystyle\frac{h(x^2)+h(2x-x^2)}{2h(x)}\le 2^{-\rho_1}.
\end{equation}
\end{lemma}

At this point, we are ready to put everything together and prove Theorem \ref{th:scalingexp}. 

\begin{proof}[Proof of Theorem \ref{th:scalingexp}]
Let us define
\begin{equation}\label{eq:defrho1}
\rho_1 = \min \left( \frac{1}{2}, -\log_2\displaystyle\sup_{\substack{x\in (0, 1), y \in [x\sqrt{2-x^2}, 2x-x^2]}}\displaystyle\frac{h(x^2)+h(y)}{2h(x)}\right),
\end{equation}
where $h(x)$ is the function of the hypothesis.  

Set 
\begin{equation}\label{eq:defalpha}
\alpha = \log_2 \left( 1+\frac{2^{-1/\mu}-2^{-\rho_1}}{2^{-1/\mu}+2^{-\rho_1}} \right).
\end{equation}
By using \eqref{eq:suph} and the fact that $\mu > 2$, we immediately realize that $2^{-1/\mu}-2^{-\rho_1} > 0$, hence that $\alpha > 0$. In addition, it easy to check that $\alpha <1$. 

Set  
\begin{equation}\label{eq:defdelta}
\delta = \frac{2^{-1/\mu}-2^{-\rho_1}}{2\sqrt{2}c_3+2^{-1/\mu}-2^{-\rho_1}},
\end{equation}
where $c_3$ is defined as in \eqref{eq:defa}. Since $2^{-1/\mu}-2^{-\rho_1} > 0$, we have that $\delta \in (0, 1)$.

In addition, $\rho_1 \le 1/2$ and the condition \eqref{eq:suphrho} clearly follows from the definition \eqref{eq:defrho1}. Consequently, we can apply Lemma \ref{lm:htoexp}, which yields formula \eqref{eq:expZn}.

Set
\begin{equation}\label{eq:defrho}
\rho = -\log_2 \left( 2^{-\rho_1} + \sqrt{2} \frac{\delta}{1-\delta}c_3\right).
\end{equation}
Then, $\rho\le \rho_1\le 1/2$, and we can apply Lemma \ref{lm:exptoBatta} with $c_1=1/\delta$, which yields
\begin{equation}\label{eq:finres}
{\mathbb P}\left(Z_n \le p_{\rm e}\hspace{0.1em} 2^{-n}\right)\ge I(W)- c_2 \hspace{0.1em} 2^{-n(\rho-\alpha)} = I(W)- c_2 \hspace{0.1em} 2^{-n/\mu},
\end{equation}
where $c_2 = \sqrt{2 \hspace{0.1em}p_{\rm e}}+2\hspace{0.1em}p_{\rm e}^{-\alpha}/\delta$ and the last equality uses the definitions \eqref{eq:defrho}, \eqref{eq:defalpha} and \eqref{eq:defdelta}. 

Consider the transmission of a polar code of block length $N=2^n$ and rate $R = I(W)- c_2 \hspace{0.1em} 2^{-n/\mu}$ over $W$. Then, by combining \eqref{eq:ubPe} and \eqref{eq:finres}, we have that the error probability under successive cancellation decoding is upper bounded by $p_{\rm e}$. Therefore, the result \eqref{eq:scalingexp} follows with $\beta_1  = c_2^{\mu}$. 

A similar proof holds for the specific case in which $W$ is a $\rm BEC$.  

\end{proof}

Now, let us briefly sketch how to prove the result stated in Remark \ref{rmk:poly}. First, we need to generalize Lemma \ref{lm:exptoBatta} by showing that, under the same hypothesis \eqref{eq:hpexBatta}, we have that, for any $\nu > 0$,
\begin{equation}\label{eq:gappoly}
{\mathbb P}\left(Z_n \le 2^{-n(\nu+1)}\right)\ge I(W)- c_4 \hspace{0.1em} 2^{-n(\rho-(\nu+1)\alpha)},
\end{equation}
where $c_4= \sqrt{2}+2\hspace{0.1em}c_1$. Then, we simply follow the procedure described in the proof of Theorem \ref{th:scalingexp} with the difference that $\alpha$ is a factor $1+\nu$ smaller than in \eqref{eq:defalpha}.

\subsection{Valid Choice for Scaling Exponent: Proof of Theorem \ref{th:validchoice}}\label{subsec:designh}

Let $W$ be a $\rm BMSC$. The proof of Theorem \ref{th:validchoice} consists in providing a good candidate for the function $h(x): [0, 1]\to [0, 1]$ such that $h(0)=h(1)=0$, $h(x)>0$ for any $x\in (0, 1)$ and \eqref{eq:suph} is satisfied with a value of $\mu$ as small as possible. In particular, we will prove that $\mu = 4.714$ is a valid choice. 

The idea is to apply repeatedly the operator $T_{\rm BMSC}$ defined in \eqref{eq:polaropgen}, until we converge to the function $h(x)$. Hence, let us define $h_k(x)$ recursively for any $k\ge 1$ as
\begin{align}
h_k(x) &= \frac{f_k(x)}{\sup_{y\in (0, 1)}f_k(y)},\label{eq:normstep}\\ 
f_k(x) &= \displaystyle\sup_{y\in [x\sqrt{2-x^2}, 2x-x^2]}\frac{h_{k-1}(x^2)+h_{k-1}(y)}{2},\label{eq:defstep}
\end{align}
with some initial condition $h_0(x)$ such that $h_0(0)=h_0(1)=0$ and $h_0(x)>0$ for any $x\in (0, 1)$. Note that the normalization step \eqref{eq:normstep} ensures that the function $h_k(x)$ does not tend to the constant function $0$ in the interval $[0, 1]$.

However, even if we choose some simple initial condition $h_0(x)$, the sequence of functions $\{h_k(x)\}_{k\in \mathbb N}$ is analytically intractable. Hence, we need to resort to numerical methods, keeping in mind that we require a \emph{provable} upper bound for any $x\in (0, 1)$ on the function
\begin{equation}\label{eq:ratioh}
r(x) = \sup_{y\in [x\sqrt{2-x^2}, 2x-x^2]} \frac{h(x^2)+h(y)}{2h(x)}.
\end{equation}
To do so, first we construct an adequate candidate for the function $h(x)$. This function will depend on some auxiliary parameters. Then, we describe an algorithm to analyze this candidate and present a choice of the parameters that gives $\mu=4.714$. Let us underline that, despite that the procedure is numerical, the resulting upper bound and the value of $\mu$ are rigorously \emph{provable}. 

For the construction part, we observe numerically that, when $k$ is large enough, the function $h_k(x)$ depends weakly on the initial condition $h_0(x)$, and it does not change much after one more iteration, i.e., $h_{k+1}(x)\approx h_k(x)$. In addition, let us point out that the goal is \emph{not} to obtain an exact approximation of the sequence of functions $\{h_k(x)\}_{k\in \mathbb N}$ defined in \eqref{eq:normstep}-\eqref{eq:defstep}. Indeed, the actual goal is to obtain a candidate $h(x)$ that satisfies \eqref{eq:suph} with a value of $\mu$ as low as possible. 

Pick a large integer $N_{\rm s}$ and let us define the sequence of functions $\{\hat{h}_k(x)\}_{k\in \mathbb N}$ as follows. For any $k\in \mathbb N$, $\hat{h}_k(x)$ is the piece-wise linear function obtained by linear interpolation from the samples $\hat{h}_k(x_i)$, where $x_i = i/N_{\rm s}$ for $i \in \{0, 1, \cdots, N_{\rm s}\}$. The samples $\hat{h}_k(x_i)$ are given by
\begin{equation}\label{eq:rechk}
\begin{split}
\hat{h}_k(x_i) &= \frac{\hat{f}_k(x_i)}{\max_{j\in \{0, 1, \cdots, N_{\rm s}\}}\hat{f}_k(x_j)},\\
\hat{f}_k(x_i) &= \frac{\hat{h}_{k-1}((x_i)^2)+\max_{j\in \{0, 1, \cdots, M_{\rm s}\}}\hat{h}_{k-1}(y_{i, j})}{2},\\
\end{split}
\end{equation}
where $M_{\rm s}$ is a large integer, and, for $j \in \{0, 1, \cdots, M_{\rm s}\}$, $y_{i, j}$ is defined as 
\begin{equation}
y_{i, j} = x_i \sqrt{2-x_i^2}+\frac{j}{M_{\rm s}}x_i\left(2-x_i-\sqrt{2-x_i^2}\right).
\end{equation}
The initial samples $\hat{h}_0(x_i)$ are obtained by evaluating at the points $\{x_i\}_{i=0}^{N_{\rm s}}$ some function $h_0(x)$ such that $h_0(0)=h_0(1)=0$ and $h_0(x)>0$ for any $x\in (0, 1)$ (see Figure \ref{fig:eigenBMSC} for a plot of $\hat{h}_0(x)$ and $\hat{h}_k(x)$). 

\begin{figure}[tb] 
\centering 
\includegraphics[width=0.8\columnwidth]{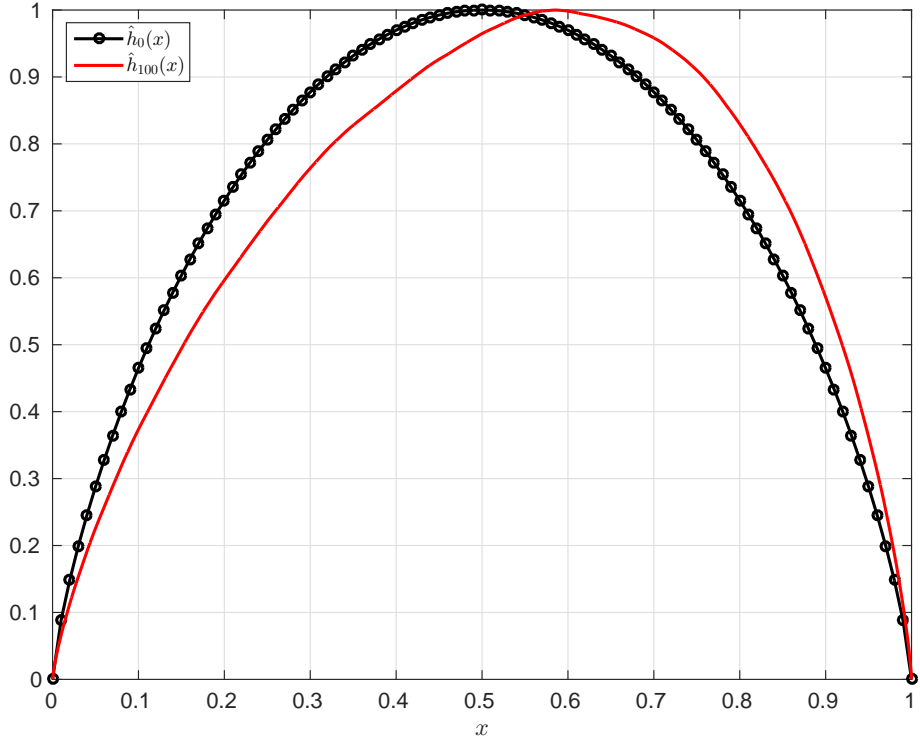}
\caption{Plot of $\hat{h}_0(x)$ (black circles) and $\hat{h}_k(x)$ (red line) after $k=100$ steps of the recursion \eqref{eq:rechk} with $N_{\rm s}=10^6$, $M_{\rm s}=10^4$, and the initial condition $f_0(x) = (x(1-x))^{3/4}$.} 
\label{fig:eigenBMSC}
\end{figure} 

It is clear that, by increasing $N_{\rm s}$ and $M_{\rm s}$, we obtain a better approximation of the sequence of functions \eqref{eq:normstep}-\eqref{eq:defstep}. In addition, by increasing $k$ we get closer to the limiting function $\lim_{k\to\infty} \hat{h}_k(x)$. Set
\begin{equation}
\hat{r}_k = \max_{i\in \{1, \cdots, N_{\rm s}-1\}}\frac{\hat{h}_{k}((x_i)^2)+\max_{j\in \{0, 1, \cdots, M_{\rm s}\}}\hat{h}_{k}(y_{i, j})}{2\hat{h}_k(x_i)}.
\end{equation}
We observe from numerical simulations that, when $k$ increases, the sequence $\hat{r}_k$ tends to the limiting value $0.86275$ for any $k$. Furthermore, this limit depends very weakly on the particular choice of the initial conditions $\{\hat{h}_0(x_i)\}_{i = 0}^{N_{\rm s}}$.

Note that, by using the samples $\{\hat{h}_k(x_i)\}_{i = 0}^{N_{\rm s}}$, $\hat{r}_k$ gives an indication of the smallest value of $\mu$ that we could hope for, i.e., $\mu = -1/\log_2 0.86275 =4.695$. Indeed, if we obtain $h(x)$ by interpolating the samples $\{\hat{h}_k(x_i)\}_{i = 0}^{N_{\rm s}}$, then $\hat{r}_k= \max_{i\in \{1, \cdots, N_{\rm s}-1\}} r(i/N_{\rm s})$, where $r(x)$ is defined in \eqref{eq:ratioh}. Therefore, $\hat{r}_k\le \sup_{x\in (0, 1)} r(x)$, i.e., $\hat{r}_k$ is a lower bound on the desired supremum, whereas we are looking for an upper bound to that quantity.

Fix a large integer $\bar{k}$ and, before computing a provable upper bound on $\sup_{x\in (0, 1)} r(x)$, let us describe the interpolation method for obtaining the candidate $h(x)$ from the samples $\{\hat{h}_{\bar{k}}(x_i)\}_{i = 0}^{N_{\rm s}}$. 

For $x$ close to $0$ and for $x$ close to $1$, linear interpolation does not yield a good candidate $h(x)$. Indeed, assume that $h(x)=\hat{h}_{\bar k}(x)$ for $x\in \left[0, 1/N_{\rm s}\right]$. Then, $\lim_{x\to 0^+} r(x) = 1$, hence $\sup_{x\in (0, 1)} r(x) \ge 1$. Similarly, if $h(x)=\hat{h}_{\bar k}(x)$ for $x\in \left[1- 1/N_{\rm s}, 1\right]$, then $\lim_{x\to 1^-} r(x) = 1$. On the contrary, if $h(x)$ grows as $x^\eta$ in a neighborhood of $0$ for $\eta\in (0, 1)$, then, it is easy to see that $\lim_{x\to 0^+} r(x) = 2^{\eta-1}$. Similarly, if $h(x)$ grows as $(1-x)^\eta$ in a neighborhood of $1$ for $\eta\in (0, 1)$, then $\lim_{x\to 1^-} r(x)=2^{\eta-1}$. Consequently, the idea is to choose $\eta$ slightly smaller than $1-1/4.695$, where $4.695$ constitutes a good approximation of the target value of $\mu$ that we want to achieve. Based on this observation, we set
\begin{align}
b_0(x) &= \hat{h}_{\bar k}\left(\frac{\bar{m}}{N_{\rm s}}\right)\hspace{0.1em} \left(\frac{\bar{m}}{N_{\rm s}}\right)^{-\eta} \hspace{0.1em} x^{\eta},\label{eq:defb0}\\
b_1(x) &= \hat{h}_{\bar k}\left(1-\frac{\bar{m}}{N_{\rm s}}\right)\hspace{0.1em} \left(\frac{\bar{m}}{N_{\rm s}}\right)^{-\eta} \hspace{0.1em} (1-x)^{\eta},\label{eq:defb1}
\end{align}
for some integer $\bar{m}\ge 2$. Then, sample $b_0(x)$ for $x\in \left[1/N_{\rm s},\bar{m}/N_{\rm s}\right]$, sample $\hat{h}_{\bar k}(x)$ for $x\in \left[\bar{m}/N_{\rm s},1-\bar{m}/N_{\rm s}\right]$, and sample $b_1(x)$ for $x\in \left[1-\bar{m}/N_{\rm s},1-1/N_{\rm s}\right]$. Note that it is better to not have a uniform sampling, but to choose the number of samples according to the rule that follows. Pick some $\delta_{\rm s}$ small enough. Then, for each couple of consecutive samples, the bigger one has to be at most a factor $1+\delta_{\rm s}$ larger than the smaller one. 

Let $\{x'_i\}_{i=1}^{N'_{\rm s}}$ denote the set of sampling positions and $\{\hat{h}_i\}_{i=1}^{N'_{\rm s}}$ denote the set of samples obtained with this procedure, where $N'_{\rm s}$ is the number of such samples. Eventually, we define the candidate $h(x)$ as
\begin{equation}
h(x) = \left\{\begin{array}{ll} b_0(x), & \mbox{ for } x \in \left[0, \displaystyle\frac{1}{N_{\rm s}}\right],\\
\\
b_1(x) & \mbox{ for } x \in \left[1- \displaystyle\frac{1}{N_{\rm s}}, 1\right],\\
\end{array}\right. 
\end{equation}
and, for $x\in \left[1/N_{\rm s}, 1-1/N_{\rm s}\right]$, $h(x)$ is obtained by linear interpolation from the samples $\{\hat{h}_i\}$. 

Concerning the analysis of $h(x)$, keep in mind that the goal is to find a provable upper bound on $\sup_{x\in (0, 1)} r(x)$. First, consider the values of $x$ in a neighborhood of $0$. The following chain of inequalities holds for any $x\in \left[0, 1/N_{\rm s}\right]$,
\begin{equation}\label{eq:ub0}
\begin{split}
r(x) &\stackrel{\mathclap{\mbox{\footnotesize(a)}}}{\le} \frac{h(x^2)+h(2x)}{2h(x)} \\
&\stackrel{\mathclap{\mbox{\footnotesize(b)}}}{\le} \frac{b_0(x^2)+b_0(2x)}{2b_0(x)} \\
&\stackrel{\mathclap{\mbox{\footnotesize(c)}}}{=} \frac{x^{\eta}}{2}+2^{\eta-1}\\
&\le  H_0 \triangleq\frac{(N_{\rm s})^{-\eta}}{2}+2^{\eta-1},
\end{split}
\end{equation}
where the inequality (a) uses that $h(y)\le h(2x)$ for any $y\in [x\sqrt{2-x^2}, 2x-x^2]$, as $h(x)$ is increasing for $x\in \left[0, 2/N_{\rm s}\right]$; the inequality (b) uses that $h(x)= b_0(x)$ for $x\in \left[0, 1/N_{\rm s}\right]$ and $h(x)\le b_0(x)$ for $x\in \left[1/N_{\rm s}, 2/N_{\rm s}\right]$, as, in that interval, $h(x)$ is the linear interpolation of samples taken from $b_0(x)$ and $b_0(x)$ is concave for any $\eta \in (0, 1)$; and the equality (c) uses the definition \eqref{eq:defb0} of $b_0(x)$. 

Second, consider the values of $x$ is a neighborhood of $1$. The following chain of inequalities holds for any $x\in \left[1- 1/N_{\rm s}, 1\right]$,
\begin{equation}
\begin{split}\label{eq:ub1}
r(x) &\stackrel{\mathclap{\mbox{\footnotesize(a)}}}{\le} \frac{h(x^2)+h(x\sqrt{2-x^2})}{2h(x)} \\
&\stackrel{\mathclap{\mbox{\footnotesize(b)}}}{\le} \frac{b_1(x^2)+b_1(x\sqrt{2-x^2})}{2b_1(x)} \\
&\stackrel{\mathclap{\mbox{\footnotesize(c)}}}{=} \frac{(1+x)^{\eta}}{2}+\frac{1}{2}\left(\frac{1-x\sqrt{2-x^2}}{1-x}\right)^{\eta}\\
&\stackrel{\mathclap{\mbox{\footnotesize(d)}}}{\le} H_1 \triangleq 2^{\eta-1} + \frac{1}{2}\left(N_{\rm s}-(N_{\rm s}-1)\sqrt{1+\frac{2}{N_{\rm s}}-\frac{1}{(N_{\rm s})^2}}\right)^{\eta}, 
\end{split}
\end{equation}
where the inequality (a) uses that $h(y)\le h(x\sqrt{2-x^2})$ for any $y\in [x\sqrt{2-x^2}, 2x-x^2]$,  as $h(x)$ is decreasing for $x\in \left[1-1/N_{\rm s}, 1\right]$; the inequality (b) uses that $h(x)= b_1(x)$ for $x\in \left[1-1/N_{\rm s}, 1\right]$ and $h(x)\le b_1(x)$ for $x\in \left[1/N_{\rm s}, 2/N_{\rm s}\right]$, as, in that interval, $h(x)$ is the linear interpolation of samples taken from $b_1(x)$ and $b_1(x)$ is concave for any $\eta \in (0, 1)$; the equality (c) uses the definition \eqref{eq:defb1} of $b_1(x)$; and the inequality (d) uses that $(1-x\sqrt{2-x^2})(1-x)^{-1}$ is decreasing for any $x\in (0, 1)$.

Finally, consider the values of $x$ in the interval $\left[1/N_{\rm s}, 1-1/N_{\rm s}\right]$. For any $i\in \{1, \cdots, N'_{\rm s}-1\}$, define
\begin{equation*}
\begin{split}
J^{+}_{i} &= \{j : x'_j\in [(x'_i)^2, (x'_{i+1})^2]\},\\
J^{-}_{i} &= \{j : x'_j\in [x'_i\sqrt{2-(x'_i)^2}, 2x'_{i+1}-(x'_{i+1})^2]\}.
\end{split}
\end{equation*}
Then, as $h(x)$ is piece-wise linear in the interval $\left[1/N_{\rm s}, 1-1/N_{\rm s}\right]$, we have that, for any $x\in [x'_i, x'_{i+1}]$, 
\begin{equation*}
\begin{split}
h(x)&\ge \min\left(h(x'_i), h(x'_{i+1})\right),\\
h(x^2)&\le h^{+}_i \triangleq \max\left(h\left((x'_i)^2\right), h\left((x'_{i+1})^2\right), \max_{j\in J^{+}_{i}}\left(h(x'_j)\right)\right),\\
\sup_{y\in [x\sqrt{2-x^2}, 2x-x^2]} h(y)&\le  h^{-}_i \triangleq\max\left(h\left(x'_i\sqrt{2-(x'_i)^2}\right), h\left(2x'_{i+1}-(x'_{i+1})^2\right), \max_{j\in J^{-}_{i}}\left(h(x'_j)\right)\right),
\end{split}
\end{equation*}
which implies that, for any $x\in [x'_i, x'_{i+1}]$,
\begin{equation}\label{eq:ubmiddle}
r(x)\le \frac{ h^{+}_i + h^{-}_i}{2\min\left(h(x'_i), h(x'_{i+1})\right)}.
\end{equation}

As a result, by combining \eqref{eq:ub0}, \eqref{eq:ub1}, and \eqref{eq:ubmiddle}, we conclude that 
\begin{equation}\label{eq:concls}
\sup_{x\in (0, 1)} r(x)\le \max\left(H_0, H_1, \max_{i\in \{1, \cdots, N'_{\rm s}-1\}}\frac{ h^{+}_i + h^{-}_i}{2\min\left(h(x'_i), h(x'_{i+1})\right)}\right), 
\end{equation}
which implies that \eqref{eq:suph} holds for any $\mu$ such that $2^{-1/\mu}$ is an upper bound on the RHS of \eqref{eq:concls}. 

Let us choose $\delta_{\rm s}$, $\eta$, the sampling positions $\{x'_i\}_{i=1}^{N'_{\rm s}}$, and the samples $\{\hat{h}_i\}_{i=1}^{N'_{\rm s}}$ to be rational numbers. Then, the RHS of \eqref{eq:concls} is the maximum of either rational numbers or sums of rational powers of rational numbers. Consequently, we can provide a provable upper bound on the RHS of \eqref{eq:concls}, hence on $\mu$. In particular, by setting $N_{\rm s}=10^6$, $M_{\rm s}=10^4$, $f_0(x) = (x(1-x))^{3/4}$, $k=100$, $\delta_{\rm s}=10^{-4}$, $\eta =78/100$, and $\bar{m}=13$, we obtain $\mu= 4.714$.

For the $\rm BEC$, the idea is to apply repeatedly the operator $T_{\rm BEC}$ defined in \eqref{eq:polarop}. Hence, by adapting the procedure described above and by setting $N_{\rm s}=10^6$, $M_{\rm s}=10^4$, $f_0(x) = (x(1-x))^{2/3}$, $k=100$, $\delta_{\rm s}=10^{-4}$, $\eta =72/100$, and $\bar{m}=5$, we obtain $\mu= 3.639$ (see Figure \ref{fig:eigenBEC} for a plot of $\hat{h}_0(x)$ and $\hat{h}_k(x)$).

\begin{figure}[tb] 
\centering 
\includegraphics[width=0.8\columnwidth]{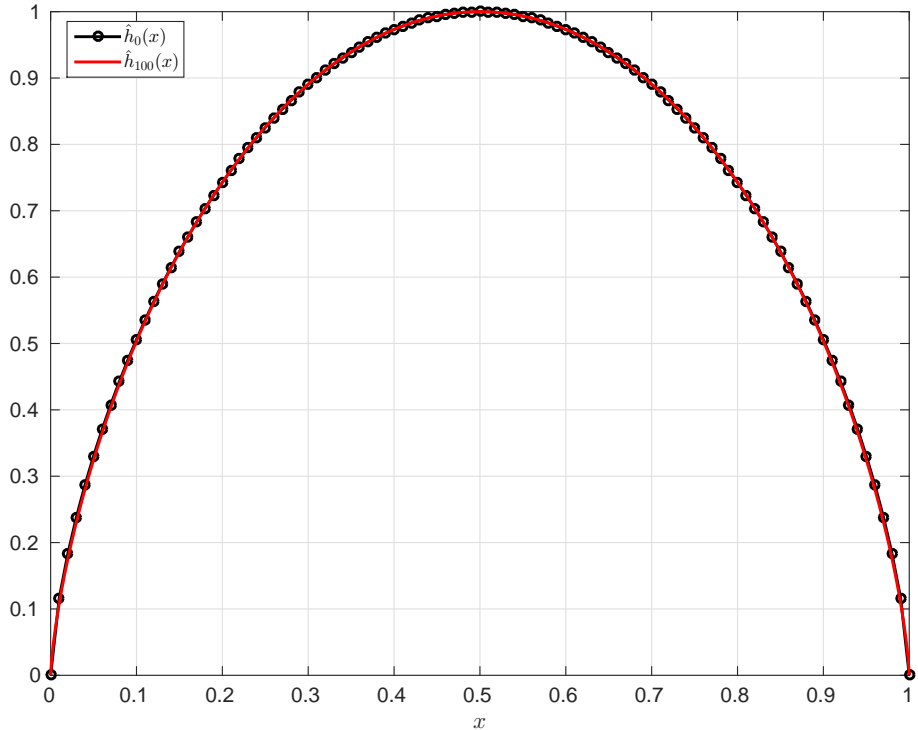}
\caption{Plot of $\hat{h}_0(x)$ (black circles) and $\hat{h}_k(x)$ (red line) after $k=100$ steps of the recursion obtained by applying the operator $T_{\rm BEC}$ defined in \eqref{eq:polarop} with $N_{\rm s}=10^6$, $M_{\rm s}=10^4$, and the initial condition $f_0(x) = (x(1-x))^{2/3}$. Differently from Figure \ref{fig:eigenBMSC}, in this case $\hat{h}_{100}(x)$ remains symmetric and very similar to the initial condition $\hat{h}_0(x)$.} 
\label{fig:eigenBEC}
\end{figure}


\section{Moderate Deviations: Joint Scaling of Error Probability and Gap to Capacity}\label{sec:joint}

The scaling exponent describes how fast the gap to capacity, as a function of the block length, tends to $0$, when the error probability is fixed. Hence, it is natural to ask how fast the gap to capacity, as a function of the block length, tends to $0$, when the error probability tends at a certain speed to $0$. The discussion of Remark \ref{rmk:poly} in Section \ref{subsec:scalstat} points out that we can let the error probability go to 0 polynomially fast in $N$, and maintain the same scaling exponent. In this section, we show that, if we allow a less favorable scaling between gap to capacity and block length (i.e. a larger scaling exponent), then the error probability goes to $0$ sub-exponentially fast in $N$. More specifically, in Section \ref{subsec:resjoint} we present the exact statement of this result together with some remarks, and in Section \ref{subsec:proofjoint} we give the proof. 

\subsection{Main Result: Statement and Discussion}\label{subsec:resjoint}

\begin{theorem}[Joint Scaling: Exponential Decay of $P_{\rm e}$]\label{th:unifexp}
Assume that there exists a function $h(x)$ that satisfies the hypotheses of Theorem \ref{th:scalingexp} for some $\mu > 2$. Consider the transmission over a $\rm BMSC$ $W$ with capacity $I(W)$ by using a polar code of rate $R<I(W)$. Then, for any $\gamma \in \left(1/(1+\mu), 1\right)$, the block length $N$ and the block error probability under successive cancellation decoding $P_{\rm e}$ are such that
\begin{equation}\label{eq:jointPeN2}
\begin{split}
P_{\rm e} &\le N\cdot 2^{-\scriptstyle N^{\scriptstyle \gamma\cdot h_2^{(-1)}\left(\frac{\gamma(\mu+1)-1}{\gamma\mu}\right)}},\\
N &\le \frac{\beta_3}{(I(W)-R)^{\mu/(1-\gamma)}},
\end{split}
\end{equation}
where $\beta_3$ is a universal constant that does not depend on $W$ or on $\gamma$, and $h_2^{(-1)}$ is the inverse of the binary entropy function defined as $h_2(x)=-x\log_2 x-(1-x)\log_2 (1-x)$ for any $x\in [0, 1/2]$. If $W$ is a $\rm BEC$, the less stringent hypothesis \eqref{eq:suphBEC} on $\mu$ is required for \eqref{eq:jointPeN2} to hold.
\end{theorem}

In short, formula \eqref{eq:jointPeN2} describes a trade-off between gap to capacity and error probability as functions of the block length $N$. Recall from Remark \ref{rmk:poly} that, if the scaling exponent is the $\mu$ given by Theorem \ref{th:validchoice}, then the error probability decays polynomially fast in $1/N$. Theorem \ref{th:unifexp} goes one step further and proves that, in order to have a faster decay of the error probability, e.g., a sub-exponential decay, it suffices to take a larger scaling exponent. 

More specifically, let $\gamma$ go from $1/(1+\mu)$ to $1$. On the one hand, the error probability goes faster and faster to $0$, since the exponent $\gamma\cdot h_2^{(-1)}\left((\gamma(\mu+1)-1)/(\gamma\mu)\right)$ is increasing in $\gamma$; on the other hand, the gap to capacity goes slower to $0$, since the exponent $\mu/(1-\gamma)$ is increasing in $\gamma$.

Before proceeding with the proof, it is useful to discuss three points. The first remark concerns the possible choices for $\mu$ in \eqref{eq:jointPeN2}. The second remark shows how to recover from Theorem \ref{th:unifexp} the result \cite{ArT09} concerning the error exponent regime. The third remark adds the Bhattacharyya parameter $Z(W)$ to the picture outlined in Theorem \ref{th:unifexp} and, in particular, it focuses on the dependency between $P_{\rm e}$ and $Z(W)$.

\begin{rmk}[Valid Choice for $\mu$ in \eqref{eq:jointPeN2}]
By constructing a function $h(x)$ as in the proof of Theorem \ref{th:validchoice} contained in Section \ref{subsec:designh}, we immediately have that valid choices of $\mu$ in \eqref{eq:jointPeN2} are $\mu=4.714$ for any $\rm BMSC$ and $\mu=3.637$ for the special case of the $\rm BEC$. 
\end{rmk}

\begin{rmk}[Error Exponent Regime and Theorem \ref{th:unifexp}]
By choosing $\gamma$ close to 1, we recover the result \cite{ArT09} concerning the error exponent regime: if we allow the gap to capacity to be arbitrary small but independent of $N$, then $P_{\rm e}$ is $O(2^{-N^\beta})$ for any $\beta\in (0, 1/2)$.\footnote{Theorem \ref{th:unifexp} also contains as a particular case the stronger result in \cite{XG13}, where the authors prove that the block length scales polynomially fast with the inverse of the gap to capacity, while the error probability is upper bounded by $2^{-N^{0.49}}$.} On the contrary, note that it is not possible to recover from Theorem \ref{th:unifexp} the result of Theorem \ref{th:scalingexp} concerning the scaling exponent regime. Indeed, choose $\gamma$ close to $1/(1+\mu)$. Then, the exponent $\gamma\cdot h_2^{(-1)}\left((\gamma(\mu+1)-1)/(\gamma\mu)\right)$ tends to $0$. This means that we approach a regime in which the error probability is independent of $N$, but $N$ is $O\left(1/(I(W)-R)^{\mu+1}\right)$, instead of $O\left(1/(I(W)-R)^{\mu}\right)$, as in \eqref{eq:scalingexp}. We believe that this is only an artifact of the proof technique used to show Theorem \ref{th:unifexp} and that it might be possible to find a joint scaling that contains as special cases the error exponent and the scaling exponent regimes. 
\end{rmk}

\begin{rmk}[Dependency between $P_{\rm e}$ and $Z(W)$]\label{rmk:Battadep}
Consider the transmission over a $\rm BMSC$ $W$ with Bhattacharyya parameter $Z(W)$. Then, under the hypotheses of Theorem \ref{th:unifexp}, it is possible to prove that
\begin{equation}\label{eq:jointPeNZ}
\begin{split}
P_{\rm e} &\le N\cdot Z(W)^{\scriptstyle  \frac{1}{2}\cdot N^{\scriptstyle \gamma\cdot h_2^{(-1)}\left(\frac{\gamma(\mu+1)-1}{\gamma\mu}\right)}},\\
N &\le \frac{\beta_4}{(I(W)-R)^{\mu/(1-\gamma)}},
\end{split}
\end{equation}
where $\beta_4$ is a universal constant that does not depend on $W$ or on $\gamma$. A sketch of the proof of this statement is given in Appendix \ref{app:Battadep}. In short, the error probability scales as $Z(W)$ raised to some power of $N$, where the exponent follows the trade-off of Theorem \ref{th:unifexp}. To see that this is a meaningful bound, consider the case of the transmission over the $\rm BEC$ in the error exponent regime. On the one hand, formula \eqref{eq:jointPeNZ} gives that $P_{\rm e}$ scales roughly as $Z(W)^{\sqrt{N}}$. On the other hand, $P_{\rm e}\ge \max_{i\in \mathcal I} Z_n^{(i)}$, where $\mathcal I$ denotes the set of information positions and $Z_n^{(i)}$ is a polynomial in $Z(W)$ with minimum degree that scales roughly\footnote{To see this, note that the minimum degree of $Z_n^{(i)}$ seen as a polynomial in $Z(W)$ is equal to the minimum distance of the code, which scales roughly as $\sqrt{N}$ according to Lemma 4 of \cite{HKU09a}.} as $\sqrt{N}$. The scaling between the error probability and the Bhattacharyya parameter will be further explored in Section \ref{sec:floors}.  
\end{rmk}

\subsection{Proof of Theorem \ref{th:unifexp}}\label{subsec:proofjoint}

\begin{proof}
Let $Z_n(W)$ be the Bhattacharyya process associated with the channel $W$. Then, by following the same procedure that gives \eqref{eq:finres}, we have that, for any $n_0\in \mathbb N$,
\begin{equation}\label{eq:bdZno}
{\mathbb P}\left(Z_{n_0} \le 2^{-n_0}\right)\ge I(W)- c_5 \hspace{0.1em} 2^{-n_0/\mu},
\end{equation}
where $c_5$ is a constant that does not depend on $n$ and is given by $c_5 = \sqrt{2} + 2/\delta$, with $\delta$ defined as in \eqref{eq:defdelta}.

Let $\{B_n\}_{n \ge 1}$ be a sequence of i.i.d. random variables with distribution Bernoulli$\left(1/2\right)$. Then, by using \eqref{eq:eqBMSC}, it is clear that, for $n\ge 1$,
\begin{equation*}
Z_{n_0+n} \le  \left\{\begin{array}{ll} Z_{n_0+n-1}^2, & \mbox{ if } B_n=1, \\ 2 Z_{n_0+n-1}, & \mbox{ if } B_n=0. \end{array}\right. 
\end{equation*}
Therefore, by applying Lemma 22 of \cite{HAU14}, we obtain that, for $n_1\ge 1$,
\begin{equation}\label{eq:lemma22}
{\mathbb P}\left(Z_{n_0+n_1} \le 2^{\scriptstyle-2^{\scriptstyle\sum_{i=1}^{n_1} B_i}}\mid Z_{n_0} = x\right)\ge 1- c_6 \hspace{0.1em} x(1-\log_2 x),
\end{equation}
with $c_6 = 2/(\sqrt{2}-1)^2$.

Consequently, we have that
\begin{equation}\label{eq:Zn0n11}
\begin{split}
{\mathbb P}\left(Z_{n_0+n_1} \le 2^{\scriptstyle-2^{\scriptstyle\sum_{i=1}^{n_1} B_i}}\right) &={\mathbb P}\left(Z_{n_0} \le 2^{-n_0}\right) \cdot {\mathbb P}\left(Z_{n_0+n_1} \le 2^{\scriptstyle-2^{\scriptstyle\sum_{i=1}^{n_1} B_i}}\mid Z_{n_0} \le 2^{-n_0}\right)  \\
& \stackrel{\mathclap{\mbox{\footnotesize(a)}}}{\ge} {\mathbb P}\left(Z_{n_0} \le 2^{-n_0}\right) \cdot \left( 1- c_6 \hspace{0.1em} 2^{-n_0}(1+n_0) \right)\\
& \stackrel{\mathclap{\mbox{\footnotesize(b)}}}{\ge} \left( I(W)- c_5 \hspace{0.1em} 2^{-n_0/\mu} \right) \cdot \left( 1- c_6\hspace{0.1em} \frac{\sqrt{2}}{\ln 2} \hspace{0.1em} 2^{-n_0/2} \right)\\
& \stackrel{\mathclap{\mbox{\footnotesize(c)}}}{\ge} I(W)- \left(c_5+c_6\hspace{0.1em} \frac{\sqrt{2}}{\ln 2}\right) \hspace{0.1em} 2^{-n_0/\mu}, 
\end{split}
\end{equation}
where the inequality (a) uses \eqref{eq:lemma22} and the fact that $1- c_6 \hspace{0.1em} x(1-\log_2 x)$ is decreasing in $x$ for any $x\le 2^{-n_0}\le 1/2$; the inequality (b) uses \eqref{eq:bdZno} and that $1- c_6 \hspace{0.1em} 2^{-n_0}(1+n_0)\ge 1- c_6 \hspace{0.1em} \sqrt{2}\cdot 2^{-n_0/2} / \ln 2 $ for any $n_0\in \mathbb N$; and the inequality (c) uses that $\mu >2$. 

Let $h_2(x)=-x\log_2 x-(1-x)\log_2 (1-x)$ denote the binary entropy function. Then, for any $\epsilon \in (0, 1/2)$,
\begin{equation}\label{eq:Zn0n12}
\begin{split}
{\mathbb P}\left(2^{\scriptstyle-2^{\scriptstyle\sum_{i=1}^{n_1} B_i}} > 2^{\scriptstyle-2^{\scriptstyle n_1 \epsilon}}\right) &= {\mathbb P}\left(\sum_{i=1}^{n_1} B_i < n_1 \epsilon\right) \\
&\le {\mathbb P}\left(\sum_{i=1}^{n_1} B_i \le \lfloor n_1 \epsilon \rfloor \right) \\
&= \sum_{k=0}^{\lfloor n_1 \epsilon \rfloor} \binom{n_1}{k} \left(\frac{1}{2}\right)^{n_1}\\
&\stackrel{\mathclap{\mbox{\footnotesize(a)}}}{\le} \left(\frac{1}{2}\right)^{n_1} 2^{n_1 h_2(\lfloor n_1 \epsilon \rfloor/n_1)} \\
&\stackrel{\mathclap{\mbox{\footnotesize(b)}}}{\le} 2^{-n_1(1-h_2(\epsilon))},
\end{split}
\end{equation}
where the inequality (a) uses formula (1.59) of \cite{RiU08}; and the inequality (b) we uses that $h_2(x)$ is increasing for any $x\le 1/2$. 

Note that, for any two events $A$ and $B$, ${\mathbb P}(A\cap B) \ge {\mathbb P}(A)+{\mathbb P}(B)-1$. Hence, by combining \eqref{eq:Zn0n11} and \eqref{eq:Zn0n12}, we obtain that
\begin{equation}\label{eq:Zn0n13}
{\mathbb P}\left(Z_{n_0+n_1} \le 2^{\scriptstyle-2^{\scriptstyle n_1\epsilon} }\right) 
\ge I(W)- \left(c_5+c_6\hspace{0.1em} \frac{\sqrt{2}}{\ln 2}\right) \hspace{0.1em} 2^{-n_0/\mu}-2^{-n_1(1-h_2(\epsilon))}.
\end{equation}
Let $n\ge 1$. Set $n_1 = \lceil\gamma n\rceil$, $n_0 = n - \lceil\gamma n\rceil$, and $\epsilon = h_2^{(-1)}\left((\gamma(\mu+1)-1)/(\gamma\mu)\right)$, where $h_2^{(-1)}(\cdot)$ is the inverse of $h_2(x)$ for any $x\in [0, 1/2]$. Note that if $\gamma \in \left(1/(1+\mu), 1\right)$, then $\epsilon\in (0, 1/2)$. Consequently, formula \eqref{eq:Zn0n13} can be rewritten as
\begin{equation}\label{eq:Zn0n1fin}
{\mathbb P}\left(Z_{n_0+n_1} \le 2^{- \scriptstyle 2^{\scriptstyle n \hspace{0.1em}\gamma\hspace{0.1em} h_2^{(-1)}\left(\frac{\gamma(\mu+1)-1}{\gamma\mu}\right)}}\right) 
\ge I(W)- c_7 \hspace{0.1em} 2^{-n\scriptstyle\frac{1-\gamma}{\mu}},
\end{equation}
with $c_7=1+\sqrt{2}\left(c_5+c_6\hspace{0.1em} \sqrt{2}/\ln 2\right)$.

Consider the transmission of a polar code of block length $N=2^n$ and rate $R$ given by the RHS of \eqref{eq:Zn0n1fin}. Then, the result \eqref{eq:jointPeN2} holds with $\beta_3 = c_7^\mu$. 

\end{proof}


\section{Absence of Error Floors}\label{sec:floors}

In the discussion of Remark \ref{rmk:Battadep} in Section \ref{subsec:resjoint}, we study the dependency between the error probability and the Bhattacharyya parameter, and we consider a setting in which, as the channel varies, the polar code used for the transmission changes accordingly. In this section, we consider a different scenario in which the polar code stays fixed as the channel varies, and we prove a result about the speed of decay of the error probability as a function of the Bhattacharyya parameter of the channel. By doing so, we conclude that polar codes are not affected by error floors. More specifically, in Section \ref{subsec:resfloor} we formalize and discuss this result, and in Section \ref{subsec:prooffloor} we present the proof.

\subsection{Main Result: Statement and Discussion}\label{subsec:resfloor}

Let $\mathcal C$ be the polar code with information set $\mathcal{I}$ designed for transmission over the $\rm BMSC$ $W'$ with Bhattacharyya parameter $Z(W')$. Then, the actual channel, over which transmission takes place, is the $\rm BMSC$ $W$ with Bhattacharyya parameter $Z(W)$. In the error floor regime, the code $\mathcal C$ is fixed and $W$ varies. The aim is to study the scaling between the error probability under SC decoding and the Bhattacharyya parameter $Z(W)$.

Denote by $Z_{n}^{(i)}(W)$ the Bhattacharyya parameter of the synthetic channel of index $i$ obtained from $W$ after $n$ steps of polarization. The main result is presented in Theorem \ref{th:floor} and it relates $Z_{n}^{(i)}(W)$ obtained from the channel $W$ to $Z_{n}^{(i)}(W')$ obtained from the channel $W'$. From this, in Corollary \ref{cor:floor}, we relate the sum of the Bhattacharyya parameters at the information positions obtained from $W$, i.e., $\tilde{P}_{\rm e}(W)\triangleq \sum_{i\in \mathcal I} Z_n^{(i)}(W)$, to the sum of Bhattacharyya parameters obtained from $W'$, i.e., $\tilde{P}_{\rm e}(W')\triangleq \sum_{i\in \mathcal I} Z_n^{(i)}(W')$. Note that the indices of the information positions are the same in both sums, since the information set $\mathcal{I}$ is fixed. The proof of Theorem \ref{th:floor} is in Section \ref{subsec:prooffloor}, and the proof of Corollary \ref{cor:floor} naturally follows. 

\begin{theorem}[Scaling of $Z_n^{(i)}(W)$]\label{th:floor}
Consider two $\rm BMSC$s $W$ and $W'$ with Bhattacharyya parameter $Z(W)$ and $Z(W')$, respectively. For $n\in \mathbb N$ and $i\in \{1, \cdots, 2^n\}$, let $Z_n^{(i)}(W)$ be the Bhattacharyya parameter of the channel $W_n^{(i)}$ obtained from $W$ via channel polarization and let $Z_n^{(i)}(W')$ be similarly obtained from $W'$. If $Z(W)\le Z(W')^2$, then
\begin{equation}\label{eq:Battarel}
Z_{n}^{(i)}(W) \le Z_{n}^{(i)}(W')^{\textstyle\frac{\log_2 Z(W)}{\log_2 Z(W')}}.
\end{equation}
If $W$ and $W'$ are $\rm BEC$s, then \eqref{eq:Battarel} holds if $Z(W)\le Z(W')$.
\end{theorem}

\begin{cor}[Scaling of $\tilde{P}_{\rm e}(W)$]\label{cor:floor}
Let $W'$ be a $\rm BMSC$ with Bhattacharyya parameter $Z(W')$ and let $\mathcal C$ be the polar code of block length $N=2^n$ and rate $R$ for transmission over $W'$. Denote by $\tilde{P}_{\rm e}(W')$ the sum of the Bhattacharyya parameters at the information positions obtained from $W'$, i.e., $\tilde{P}_{\rm e}(W')\triangleq \sum_{i\in \mathcal I} Z_n^{(i)}(W')$, where $\mathcal{I}$ is the information set of the polar code $\mathcal C$. Now, consider the transmission over the $\rm BMSC$ $W$ with Bhattacharyya parameter $Z(W)$ by using the polar code $\mathcal C$ and let $\tilde{P}_{\rm e}(W)$ be the sum of the Bhattacharyya parameters at the information positions obtained from $W$, i.e., $\tilde{P}_{\rm e}(W)\triangleq \sum_{i\in \mathcal I} Z_n^{(i)}(W)$. If $Z(W)\le Z(W')^2$, then
\begin{equation}\label{eq:Perel}
\tilde{P}_{\rm e}(W)\le \tilde{P}_{\rm e}(W')^{\textstyle\frac{\log_2 Z(W)}{\log_2 Z(W')}}.
\end{equation}
If $W$ and $W'$ are $\rm BEC$s, then \eqref{eq:Perel} holds if $Z(W)\le Z(W')$.
\end{cor}

Now, let us discuss how the results above imply that polar codes are not affected by error floors. Denote by $P_{\rm e}(W)$ the error probability under SC decoding for transmission of $\mathcal{C}$ over $W$ and recall from \eqref{eq:ubPe} that $P_{\rm e}(W)\le\tilde{P}_{\rm e}(W)$. Hence, formula \eqref{eq:Perel} implies that
\begin{equation}\label{eq:ubpez1}
P_{\rm e}(W)\le Z(W)^{\textstyle\frac{\log_2\tilde{P}_{\rm e}(W')}{\log_2 Z(W')}}.
\end{equation}
Note that the upper bound \eqref{eq:jointPeNZ} on $P_{\rm e}$ comes from an identical upper bound on the sum of the Bhattacharyya parameters $\tilde{P}_{\rm e}$. Thus, by choosing $\gamma \approx 1$ in \eqref{eq:jointPeNZ}, we have that $\tilde{P}_{\rm e}(W')$ scales roughly as $Z(W')^{\sqrt{N}}$. Therefore, from \eqref{eq:ubpez1} we conclude that $P_{\rm e}(W)$ scales roughly as $Z(W)^{\sqrt{N}}$. This fact excludes the existence of an error floor region. 

Furthermore, in the discussion of Remark \ref{rmk:Battadep}, we pointed out that $P_{\rm e}(W)$ scales as $Z(W)^{\sqrt{N}}$ when $W$ is fixed and, consequently, the polar code can be constructed according to the actual transmission channel. Whereas, in the error floor regime, we fix a polar code and let the transmission channel vary, which means that the code cannot depend on the transmission channel. Hence, from the discussion above, it follows that the dependency between the error probability and the Bhattacharyya parameter of the channel is essentially the same as in the case in which we design the polar code for the actual transmission channel. As a result, in terms of this particular scaling, nothing is lost by considering a ``mismatched'' code. However, considering a  ``mismatched'' code yields a loss in rate. Indeed, if $W$ and $W'$ are $\rm BEC$s, then \eqref{eq:capBatta0} holds with equality, and $Z(W)\le Z(W')$ implies that $I(W)\ge I(W')$. If $W$ and $W'$ can be any $\rm BMSC$, by using \eqref{eq:capBatta0} and \eqref{eq:capBatta} we easily deduce that $Z(W)\le Z(W')^2$ implies $I(W)\ge I(W')$. Recall that the rate of a polar code for $W'$ is such that $R<I(W')$, and the rate of a polar code for $W$ is such that $R<I(W)$. As $I(W)\ge I(W')$, by constructing a polar code for $W$, we can transmit reliably at larger rates.

Before proceeding with the proof of Theorem \ref{th:floor}, let us make a brief remark concerning the case $Z(W)\in \left(Z(W')^2, Z(W')\right]$.

\begin{rmk}[The case $Z(W)\in(Z(W')^2, Z(W'){]}$]\label{rmk:badcase}
If $W$ and $W'$ are $\rm BEC$s, then \eqref{eq:Battarel} and \eqref{eq:Perel} hold for any $Z(W)\le Z(W')$, i.e., for the whole range of parameters of interest, as we think of $W$ as a ``better'' channel than $W'$. On the contrary, if $W$ and $W'$ can be any $\rm BMSC$, we require that $Z(W)\le Z(W')^2$. If there is no additional hypothesis on $W$ and $W'$, the main result \eqref{eq:Battarel} cannot hold in the case $Z(W)\in(Z(W')^2, Z(W'){]}$. Indeed, if $Z(W)=Z(W')$, we can choose $W$ and $W'$ such that $I(W)<I(W')$. If $I(W)<I(W')$, then the number of indices $i_1$ such that $\lim_{n\to \infty}Z_{n}^{(i_1)}(W)=0$ is smaller than the number of indices $i_2$ such that $\lim_{n\to \infty}Z_{n}^{(i_2)}(W')=0$. Hence, \eqref{eq:Battarel} cannot hold for any $i\in \{1, \cdots, 2^n\}$. A natural additional hypothesis consists in assuming that $W'$ is degraded with respect to $W$, i.e., $W \succ W'$. In this case, we can at least ensure that $Z_{n}^{(i)}(W) \le Z_{n}^{(i)}(W')$. However, it is possible to find $W$ and $W'$ such that \eqref{eq:Battarel} is violated for $n=1$ when $Z(W)\in(Z(W')^2, Z(W'){]}$. We leave as open questions whether the bound \eqref{eq:Perel} is still valid and what kind of looser bound holds, when $W \succ W'$ and $Z(W)\in(Z(W')^2, Z(W'){]}$.
\end{rmk}

\subsection{Proof of Theorem \ref{th:floor}}\label{subsec:prooffloor}

\begin{proof}
Assume that, for any $j\in \{1, \cdots, 2^{n-1}\}$ and for some $\eta\in \mathbb R^+$,
\begin{equation}\label{eq:assumpt}
Z_{n-1}^{(j)}(W) \le Z_{n-1}^{(j)}(W')^{\eta}.
\end{equation}
Then, let us study for what values of $\eta$ we have that \eqref{eq:assumpt} implies that, for any $i\in \{1, \cdots, 2^{n}\}$,
\begin{equation}\label{eq:newres}
Z_{n}^{(i)}(W) \le Z_{n}^{(i)}(W')^{\eta}.
\end{equation}

Recall, from Section \ref{sec:prel}, that $(b_1^{(i)}, \cdots, b_n^{(i)})$ denotes the binary representation of the integer $i-1$ over $n$ bits. Let $i$ be an even integer and set $i^+=\displaystyle\frac{i}{2}$. Then, $b_n^{(i)}=1$ and the binary representation of $i^+-1$ over $n-1$ bits is $(b_1^{(i)}, \cdots, b_{n-1}^{(i)})$. Hence, the following chain of inequalities holds for any $\rm BMSC$ $W$:
\begin{equation}\label{eq:ineqplus}
\begin{split}
Z_{n}^{(i)}(W) &\stackrel{\mathclap{\mbox{\footnotesize(a)}}}{=} \left(Z_{n-1}^{(i^+)}(W)\right)^2 \\
&\stackrel{\mathclap{\mbox{\footnotesize(b)}}}{\le} \left(Z_{n-1}^{(i^+)}(W')\right)^{2\eta} \\
&\stackrel{\mathclap{\mbox{\footnotesize(c)}}}{=} \left(Z_{n}^{(i)}(W')\right)^{\eta},
\end{split}
\end{equation}
where the equality (a) uses \eqref{eq:defWni} and \eqref{eq:plusB}; the inequality (b) uses the assumption \eqref{eq:assumpt} with $j=i^+$; and the equality (c) uses again \eqref{eq:defWni} and \eqref{eq:plusB}. Consequently, if $i$ is even, then \eqref{eq:newres} holds for any $\rm BMSC$ $W$ without any restriction on $\eta$. 

Let $i$ be an odd integer and set $i^-=\displaystyle\frac{i-1}{2}$. Then, $b_n^{(i)}=0$ and the binary representation of $i^--1$ over $n-1$ bits is $(b_1^{(i)}, \cdots, b_{n-1}^{(i)})$. Hence, the following chain of inequalities holds for any $\rm BMSC$ $W$:
\begin{equation}\label{eq:ineqminus}
\begin{split}
Z_{n}^{(i)}(W)&\stackrel{\mathclap{\mbox{\footnotesize(a)}}}{\le} Z_{n-1}^{(i^-)}(W)\left(2-Z_{n-1}^{(i^-)}(W)\right) \\
&\stackrel{\mathclap{\mbox{\footnotesize(b)}}}{\le}\left(Z_{n-1}^{(i^-)}(W')\right)^{\eta}\left(2-\left(Z_{n-1}^{(i^-)}(W') \right)^{\eta}\right)\\
& \stackrel{\mathclap{\mbox{\footnotesize(c)}}}{\le} \left(Z_{n-1}^{(i^-)}(W')\right)^{\eta} \left( 2-\left(Z_{n-1}^{(i^-)}(W')\right)^2\right)^{\eta/2} \\
& \stackrel{\mathclap{\mbox{\footnotesize(d)}}}{\le} \left(Z_{n}^{(i)}(W')\right)^{\eta},
\end{split}
\end{equation}
where the inequality (a) uses \eqref{eq:defWni} and \eqref{eq:minusB}; the inequality (b) uses the assumption \eqref{eq:assumpt} with $j=i^-$; the inequality (c) uses that $2-x^\eta \le (2-x^2 )^{\eta/2}$ for any $x\in [0, 1]$ if and only if $\eta\ge 2$; and the inequality (d) uses again \eqref{eq:defWni} and \eqref{eq:minusB}. Consequently, if $i$ is odd, then \eqref{eq:newres} holds for any $\rm BMSC$ $W$, provided that $\eta\ge 2$. If $W$ is a $\rm BEC$, a less restrictive condition on $\eta$ is necessary. Indeed, the following chain of inequalities holds when $W$ is a $\rm BEC$:
\begin{equation}\label{eq:ineqminusBEC}
\begin{split}
Z_{n}^{(i)}(W) &\stackrel{\mathclap{\mbox{\footnotesize(a)}}}{=} Z_{n-1}^{(i^-)}(W)\left(2-Z_{n-1}^{(i^-)}(W)\right) \\
&\stackrel{\mathclap{\mbox{\footnotesize(b)}}}{\le} \left(Z_{n-1}^{(i^-)}(W')\right)^{\eta}\left(2-\left(Z_{n-1}^{(i^-)}(W') \right)^{\eta}\right)\\
& \stackrel{\mathclap{\mbox{\footnotesize(c)}}}{\le} \left(Z_{n-1}^{(i^-)}(W')\right)^{\eta}\left(2-Z_{n-1}^{(i^-)}(W')\right)^{\eta} \\
& \stackrel{\mathclap{\mbox{\footnotesize(d)}}}{=} \left(Z_{n}^{(i)}(W')\right)^{\eta},
\end{split}
\end{equation}
where the equality (a) uses \eqref{eq:defWni} and \eqref{eq:minusBEC}; the inequality (b) uses the assumption \eqref{eq:assumpt} with $j=i^-$; the inequality (c) uses that $2-x^\eta \le (2-x)^{\eta}$ for any $x\in [0, 1]$ if and only if $\eta\ge 1$; and the equality (d) uses again \eqref{eq:defWni} and \eqref{eq:minusBEC}. Consequently, if $i$ is odd and $W$ is a $\rm BEC$, then \eqref{eq:newres} holds provided that $\eta\ge 1$.

By combining \eqref{eq:ineqplus} and \eqref{eq:ineqminus}, we have that if \eqref{eq:assumpt} holds for $\eta \ge 2$ after $n-1$ steps of polarization, then the same relation holds for $\eta \ge 2$ after $n$ steps of polarization. This means that the inequality stays preserved after one more step of polarization. Clearly, as the Bhattacharyya parameter is between $0$ and $1$, a smaller value of $\eta$ gives a tighter bound. Since $Z_0^{(1)}(W) = Z(W)$ and $Z_0^{(1)}(W') = Z(W')$, the smallest choice for $\eta$ is $\log_2 Z(W)/\log_2 Z(W')$. The condition $\eta\ge 2$ is equivalent to $Z(W)\le Z(W')^2$ and, for the case of the $\rm BEC$, the condition $\eta\ge 1$ is equivalent to $Z(W)\le Z(W')$. Eventually, the result \eqref{eq:Battarel} follows easily by induction. 

\end{proof}


\section{Concluding Remarks}\label{sec:concl}

In this paper, we have presented a unified view on the scaling of polar codes, by studying the relation among the fundamental parameters at play, i.e., the block length $N$, the rate $R$, the error probability under successive cancellation (SC) decoding $P_{\rm e}$, the capacity of the transmission channel $I(W)$ and its Bhattacharyya parameter $Z(W)$. Here, we summarize the main results contained in this work, along with open questions and directions for future research.

First of all, we have proved a new upper bound on the scaling exponent for any $\rm BMSC$ $W$. The setting is the following: we fix the error probability $P_{\rm e}$ and we study how the gap to capacity $I(W)-R$ scales with the block length $N$. In particular, $N$ is $O\left(1/(I(W)-R)^\mu \right)$, where $\mu$ is the so-called scaling exponent whose value depends on $W$, and we show a better upper bound on $\mu$ valid for any $\rm BMSC$ $W$. The proof technique consists in relating the value of $\mu$ to the supremum of a function that fulfills certain constraints. Then, we upper bound the supremum by constructing and analyzing a suitable candidate function. We underline that the proposed bound is \emph{provable} and that the analysis of the algorithm is not affected by numerical errors, as all the computations can be reduced to computations over integers, thus they can be performed exactly. The proposed proof technique yields $\mu\le 4.714$ for any $\rm BMSC$, which essentially improves by $1$ the existing upper bound. If $W$ is a $\rm BEC$, we obtain $\mu\le 3.639$, which approaches the value previously computed with heuristic methods. These bounds can be slightly tightened simply by increasing the number of samples used by the algorithm. 
Possibly the most interesting challenge concerning the performance of polar codes consists in improving the scaling exponent, i.e., the speed of decay of the gap to capacity, by changing the construction of the code and by devising better decoding algorithms. One promising method consists in constructing a code that interpolates between a polar and a Reed-Muller code and in using the MAP decoder, or even the low-complexity SCL decoder \cite{MoHaUr14}. Another possibility is to consider the polarization of general $q\times q$ kernels, as briefly discussed at the end of this section. 

Second, we have considered a moderate deviations regime and proved a trade-off between the speed of decay of the error probability and that of the gap to capacity. The setting is the following: we do not fix either the error probability $P_{\rm e}$ or the gap to capacity $I(W)-R$, but we study how fast both $P_{\rm e}$ and $I(W)-R$, as functions of the block length $N$, go to $0$ at the same time. In particular, we show that, if the gap to capacity is such that $$N=O\left(\frac{1}{(I(W)-R)^{\mu/(1-\gamma)}} \right),\qquad \mbox{ for } \gamma \in \left(\frac{1}{(1+\mu)}, 1\right),$$ then the error probability is given by $$P_{\rm e}=O\left(N\cdot 2^{-\scriptstyle N^{\scriptstyle \gamma\cdot h_2^{(-1)}\left(\frac{\gamma(\mu+1)-1}{\gamma\mu}\right)}}\right).$$ Note that, as the exponents $\mu/(1-\gamma)$ and $\gamma\cdot h_2^{(-1)}\left((\gamma(\mu+1)-1)/(\gamma\mu)\right)$ are both increasing in $\gamma$, if the error probability decays faster, then the gap to capacity decays slower. This trade-off recovers the existing result for the error exponent regime, but it does not match the new bound on the scaling exponent. An interesting open question consists in finding the optimal trade-off that provides the fastest possible decay of the error probability, given a certain speed of decay of the gap to capacity. Note that this optimal trade-off would match the existing results for both the error exponent and the scaling exponent regimes.

Third, we have proved that polar codes are not affected by error floors. The setting is the following: we fix a polar code of block length $N$ and rate $R$ designed for a channel $W'$, we let the transmission channel $W$ vary, and we study how the error probability $P_{\rm e}(W)$ scales with the Bhattacharyya parameter $Z(W)$ of the channel $W$. In particular, we show that $$P_{\rm e}(W)\le Z(W)^{\textstyle\frac{\log_2 \tilde{P}_{\rm e}(W')}{\log_2 Z(W')}},$$ where $\tilde{P}_{\rm e}(W')$ denotes the sum of the Bhattacharyya parameters at the information positions obtained by polarizing $W'$. In addition, $\log_2 \tilde{P}_{\rm e}(W')/\log_2 Z(W')$ scales roughly as $\sqrt{N}$, which is the best possible scaling according to the error exponent regime. Hence, the scaling between $P_{\rm e}$ and $Z(W)$ would have been the same, even if we ``matched'' the code to the channel. However, when $W$ and $W'$ can be any $\rm BMSC$, the result holds only if $Z(W)\le Z(W')^2$. An interesting open question is to explore further the case $Z(W)\in (Z(W')^2, Z(W')]$, in order to see whether a similar but perhaps less tight bound still holds. 

Finally, let us highlight that the technical tools developed in this paper have proven useful also in different scenarios. Indeed, the analysis of Section \ref{sec:scalexp} is the starting point for the characterization of the scaling exponent of binary-input energy-harvesting channels \cite{FT16} and of $q$-ary polar codes based on $q\times q$ Reed-Solomon polarization kernels \cite{PU16}.

Why are we interested in $q\times q$ kernels? Such kernels have the potential to improve the scaling behavior of polar codes. As for the error exponent, in \cite{KoSaUrb10} it is proved that, as $q$ goes large, the error probability scales roughly as $2^{-N}$. As for the scaling exponent, in \cite{F14} it is observed that $\mu$ can be reduced when $q \ge 8$. In the recent paper \cite{PU16}, it is shown that, for transmission over the erasure channel, the optimal scaling exponent $\mu = 2$ is approached by using a large kernel and, at the same time, a large alphabet. Furthermore, in \cite{HHthesis}, the author gives evidence supporting the conjecture that, in order to obtain $\mu=2$, it suffices to consider a large random kernel over a binary alphabet. Therefore, providing a rigorous proof of such a conjecture is a very interesting open problem.

\appendix

\subsection{Proof of Lemma \ref{lm:exptoBatta}} \label{app:exptoBatta}

\begin{proof}
First of all, we upper bound ${\mathbb P}(Z_n \in \left[p_{\rm e}\hspace{0.1em} 2^{-n}, 1-p_{\rm e}\hspace{0.1em} 2^{-n}\right])$ as follows: 
\begin{equation}\label{eq:B}
\begin{split}
{\mathbb P}\left(Z_n \in \left[p_{\rm e}\hspace{0.1em} 2^{-n}, 1-p_{\rm e}\hspace{0.1em} 2^{-n}\right]\right)
& \stackrel{\mathclap{\mbox{\footnotesize(a)}}}{=} {\mathbb P}\left( (Z_n(1-Z_n))^{\alpha}\ge (p_{\rm e}\hspace{0.1em} 2^{-n}(1-p_{\rm e}\hspace{0.1em} 2^{-n}))^{\alpha}\right)\\
& \stackrel{\mathclap{\mbox{\footnotesize(b)}}}{\le} \frac{{\mathbb E}\left[(Z_n(1-Z_n))^{\alpha}\right]}{(p_{\rm e}\hspace{0.1em} 2^{-n}(1-p_{\rm e}\hspace{0.1em} 2^{-n}))^{\alpha}}\\
&\stackrel{\mathclap{\mbox{\footnotesize(c)}}}{\le} \frac{c_1\hspace{0.1em}2^{-n\rho}}{(p_{\rm e}\hspace{0.1em} 2^{-n}(1-p_{\rm e}\hspace{0.1em} 2^{-n}))^{\alpha}}\\
&\stackrel{\mathclap{\mbox{\footnotesize(d)}}}{\le} 2\hspace{0.1em}c_1\hspace{0.1em} p_{\rm e}^{-\alpha} \hspace{0.1em} 2^{-n(\rho-\alpha)},
\end{split}
\end{equation}
where the equality (a) uses the concavity of the function $f(x)=(x(1-x))^\alpha$; the inequality (b) follows from Markov inequality; the inequality (c) uses the hypothesis ${\mathbb E}[(Z_n(1-Z_n))^{\alpha}]\le c_1\hspace{0.1em}2^{-n\rho}$; and the inequality (d) uses that $1-p_{\rm e}\hspace{0.1em}2^{-n}\ge 1/2$ for any $n\ge 1$. 

Let us define
\begin{equation}\label{eq:defABC}
\begin{split}
A &= {\mathbb P}\left(Z_n \in \left[0, p_{\rm e}\hspace{0.1em} 2^{-n}\right)\right),\\ 
B &= {\mathbb P}\left(Z_n \in \left[p_{\rm e}\hspace{0.1em} 2^{-n}, 1-p_{\rm e}\hspace{0.1em} 2^{-n}\right]\right),\\ 
C &= {\mathbb P}\left(Z_n \in \left(1-p_{\rm e}\hspace{0.1em} 2^{-n}, 1\right]\right),\\ 
\end{split}
\end{equation}
and let $A'$, $B'$, and $C'$ be the fraction of $A$, $B$, and $C$, respectively, that will go to $0$ as $n\to \infty$. More formally, 
\begin{equation}\label{eq:Aprime}
\begin{split}
A' &= \liminf_{m\to \infty} {\mathbb P}\left(Z_n \in \left[0, p_{\rm e}\hspace{0.1em} 2^{-n}\right), Z_{n+m}\le 2^{-m}\right),\\ 
B' &= \liminf_{m\to \infty} {\mathbb P}\left(Z_n \in \left[p_{\rm e}\hspace{0.1em} 2^{-n}, 1-p_{\rm e}\hspace{0.1em} 2^{-n}\right], Z_{n+m}\le 2^{-m}\right),\\ 
C' &= \liminf_{m\to \infty} {\mathbb P}\left(Z_n \in \left(1-p_{\rm e}\hspace{0.1em} 2^{-n}, 1\right], Z_{n+m}\le 2^{-m}\right).\\ 
\end{split}
\end{equation} 
In \eqref{eq:Aprime} we simply require that $Z_{n+m}$ goes to $0$ as $m$ goes large, and we do not have any requirement on the speed at which it does so. Hence, we could substitute $2^{-m}$ with any other function that is $O(2^{-2^{\beta m}})$ for any $\beta\in (0, 1/2)$, see \cite{ArT09}.

It is clear that
\begin{equation}\label{eq:sumABC}
A' + B' +C' = \liminf_{m\to \infty} {\mathbb P}\left(Z_{n+m}\le 2^{-m}\right)=I(W).
\end{equation}
In addition, from \eqref{eq:B}, we have that
\begin{equation}\label{eq:Bprime}
B' \le B \le 2\hspace{0.1em}c_1\hspace{0.1em} p_{\rm e}^{-\alpha} \hspace{0.1em} 2^{-n(\rho-\alpha)}.
\end{equation}
In order to upper bound $C'$, we proceed as follows:
\begin{equation}\label{eq:Cprimecalc}
\begin{split}
C' &= \liminf_{m\to \infty} {\mathbb P}\left(Z_{n+m}\le 2^{-m}\mid Z_n \in \left(1-p_{\rm e}\hspace{0.1em} 2^{-n}, 1\right]\right) \cdot {\mathbb P}\left(Z_n \in \left(1-p_{\rm e}\hspace{0.1em} 2^{-n}, 1\right]\right) \\
&\le \liminf_{m\to \infty} {\mathbb P}\left(Z_{n+m}\le 2^{-m}\mid Z_n \in \left(1-p_{\rm e}\hspace{0.1em} 2^{-n}, 1\right]\right).
\end{split}
\end{equation}
The last term equals the capacity of a channel with Bhattacharyya parameter in the interval $\left(1-p_{\rm e}\hspace{0.1em} 2^{-n}, 1\right]$. Using \eqref{eq:capBatta}, we obtain that 
\begin{equation}\label{eq:Cprime}
C' \le \sqrt{1-(1-p_{\rm e}\hspace{0.1em} 2^{-n})^2}\le \sqrt{2\hspace{0.1em} p_{\rm e} \hspace{0.1em}2^{-n}}. 
\end{equation}  
As a result, we have that
\begin{equation*}
\begin{split}
{\mathbb P}\left(Z_n \in \left[0, p_{\rm e}\hspace{0.1em} 2^{-n}\right)\right) &= A \ge A' \\
&\stackrel{\mathclap{\mbox{\footnotesize(a)}}}{=} I(W)- B' - C' \\
&\stackrel{\mathclap{\mbox{\footnotesize(b)}}}{\ge} I(W) - 2\hspace{0.1em}c_1\hspace{0.1em} p_{\rm e}^{-\alpha} \hspace{0.1em} 2^{-n(\rho-\alpha)} - \sqrt{2\hspace{0.1em} p_{\rm e}\hspace{0.1em}2^{-n}},\\
&\stackrel{\mathclap{\mbox{\footnotesize(c)}}}{\ge} I(W) - \left(\sqrt{2 \hspace{0.1em}p_{\rm e}}+2\hspace{0.1em}c_1\hspace{0.1em}p_{\rm e}^{-\alpha}\right)2^{-n(\rho-\alpha)},
\end{split}
\end{equation*}
where the equality (a) uses \eqref{eq:sumABC}; the inequality (b) uses \eqref{eq:Bprime} and \eqref{eq:Cprime}; and the inequality (c) uses that $\rho\le 1/2$. This chain of inequalities implies the desired result.

\end{proof}

\subsection{Proof of Lemma \ref{lm:htoexp}} \label{app:htoexp}

\begin{proof}
Let $\alpha^* =\min(1/2, \rho_1/\log_2(4/3))$. As ${\mathbb E}\left[\bigl( Z_n(1-Z_n) \bigr)^{\alpha}\right]$ is decreasing in $\alpha$, we can assume that $\alpha < \alpha^*$ without loss of generality. As $h(x)\ge 0$ for any $x \in [0, 1]$ and $Z_n \in [0, 1]$ for any $n \in \mathbb N$, we have that 
\begin{equation}\label{eq:fin1}
{\mathbb E}\left[ \bigl( Z_n(1-Z_n) \bigr)^{\alpha}\right] \le \frac{1}{\delta}{\mathbb E}\left[(1-\delta)h(Z_n)+\delta(Z_n(1-Z_n))^{\alpha}\right]=\frac{1}{\delta}{\mathbb E}\left[g(Z_n)\right],
\end{equation}
with
\begin{equation}\label{eq:defgx}
g(x) = (1-\delta)h(x)+\delta (x(1-x))^{\alpha}. 
\end{equation}

Let
\begin{equation*}
L_g = \sup_{x\in (0, 1), y\in [x\sqrt{2-x^2}, 2x-x^2]} \frac{g(x^2)+g(y)}{2g(x)}.
\end{equation*}
Then, by definition \eqref{eq:eqBMSC} of the Bhattacharyya process $Z_n$, we have that
\begin{equation*}
{\mathbb E}\left[g(Z_n)\mid Z_{n-1}\right] \le g(Z_{n-1}) L_g.
\end{equation*}
Consequently, by induction, one can readily prove that 
\begin{equation}\label{eq:fin3}
{\mathbb E}\left[g(Z_n)\right] \le (L_g)^n \hspace{0.1em} g(Z(W)) \le (L_g)^n, 
\end{equation}
where the last inequality follows from the fact that $g(x) \leq 1$ for $x \in [0,1]$.

Now, by combining \eqref{eq:fin1} with \eqref{eq:fin3}, we obtain that
\begin{equation} \label{Z_n-L_g} 
{\mathbb E}[( Z_n(1-Z_n))^{\alpha}] \le \frac{1}{\delta}(L_g)^n.
\end{equation}
Hence, in order to conclude the proof, it remains to find an upper bound on $L_g$, i.e., to show that $L_g \le 2^{-\rho_1} + 2\sqrt{2} \delta c_3$.
By using \eqref{eq:suphrho}, after some calculations, we have that
\begin{equation}\label{eq:supexpr}
\frac{g(x^2)+g(y)}{2g(x)} \le \frac{(1-\delta)h(x)2^{-\rho_1}+\displaystyle\frac{\delta}{2} \Bigl(\left(x^2(1-x)(1+x)\right)^{\alpha}+(y(1-y))^{\alpha}\Bigr)}{(1-\delta)h(x) + \delta(x(1-x))^{\alpha}}.
\end{equation}
For any $y\in [x\sqrt{2-x^2}, 2x-x^2]$, we obtain
\begin{equation}\label{eq:y1}
y(1-y)\le x(2-x)(1-x\sqrt{2-x^2}).
\end{equation}
In addition, for any $x\in (0, 1)$,
\begin{equation}\label{eq:y2}
1-x\sqrt{2-x^2} \le (1-x)^{4/3}.
\end{equation}
In order to prove \eqref{eq:y2}, one strategy is the following: elevate the LHS and the RHS to the third power; isolate on one side the terms that multiply $\sqrt{2-x^2}$; and square again the LHS and the RHS. In this way, we have that \eqref{eq:y2} is equivalent to 
\begin{equation*}
(1-x)^4 (2+8x+3x^2+4x^3-4x^4-4x^5-x^6)\ge 0,
\end{equation*}
which is satisfied when $x \in (0, 1)$. 

Therefore, by combining \eqref{eq:supexpr}, \eqref{eq:y1}, and \eqref{eq:y2}, we obtain that
\begin{equation}\label{eq:supexpr2}
\frac{g(x^2)+g(y)}{2g(x)} \le \frac{(1-\delta)h(x)2^{-\rho_1}+\delta (x(1-x))^{\alpha}\hspace{0.1em}t(x)}{(1-\delta)h(x) + \delta(x(1-x))^{\alpha}},
\end{equation}
with
\begin{equation}\label{eq:tx}
t(x)=\displaystyle\frac{1}{2}\left( \bigl(x(1+x)\bigr)^\alpha + \bigl((2-x)(1-x)^{1/3}\bigr)^\alpha \right).
\end{equation}

First of all, we upper bound the expression on the RHS of \eqref{eq:supexpr2} when $x$ is small. Clearly, $t(0)< 2^{-\rho_1}$ and $t(1/2) > 2^{-\rho_1}$, as $\rho_1 \le 1/2$ and $\alpha < \alpha^*$. In addition, some passages of calculus show that the second derivative of $t(x)$ is given by 
\begin{equation*}
\frac{\alpha}{2}\frac{(x(1+x))^{\alpha}}{x^2 (1+x)^2}\left(-1-2x-2x^2+\alpha(1+2x)^2\right) + \frac{\alpha}{18}\frac{\bigl((2-x)(1-x)^{1/3}\bigr)^{\alpha}}{(2-3x+x^2)^2}\left(-21+30x-12x^2+\alpha(5-4x)^2\right).
\end{equation*}
As $\alpha < 1/2$, we have that
\begin{equation}
\begin{split}
-1-2x-2x^2+\alpha(1+2x)^2 &\le -1-2x-2x^2+\frac{(1+2x)^2}{2} < 0,\\
-21+30x-12x^2+\alpha(5-4x)^2 &\le -1-2x-2x^2+\frac{(5-4x)^2}{2} < 0.\\
\end{split}
\end{equation}
Hence, $t(x)$ is concave for any $x\in (0, 1)$. This implies that there exist $\epsilon_1(\alpha), \epsilon_2(\alpha) \in (0, 1)$ such that
\begin{equation}\label{eq:t}
t(x) \le 2^{-\rho_1}, \qquad \qquad \forall\, x\in [0, \epsilon_1(\alpha)]\cup [1-\epsilon_2(\alpha), 1]. 
\end{equation}
Indeed, the precise values of $\epsilon_1(\alpha)$ and $\epsilon_2(\alpha)$ can be found from \eqref{eps-solution}. 
By combining \eqref{eq:supexpr2} with \eqref{eq:t}, we have that, for any $x\in [0, \epsilon_1(\alpha)]\cup [1-\epsilon_2(\alpha), 1]$ and for any $y\in [x\sqrt{2-x^2}, 2x-x^2]$,
\begin{equation}\label{eq:bLg1}
\frac{g(x^2)+g(y)}{2g(x)} \le 2^{-\rho_1}.
\end{equation}

Then, we upper bound the expression on the RHS of \eqref{eq:supexpr2} when $x$ is not too small, namely, $x\in (\epsilon_1(\alpha), 1-\epsilon_2(\alpha))$: 
\begin{equation}\label{eq:bLg2}
\begin{split}
\frac{(1-\delta)h(x)2^{-\rho_1}+\delta (x(1-x))^{\alpha}\hspace{0.1em}t(x)}{(1-\delta)h(x) + \delta(x(1-x))^{\alpha}}
& \stackrel{\mathclap{\mbox{\footnotesize(a)}}}{\le} \frac{(1-\delta)h(x)2^{-\rho_1}+\delta (x(1-x))^{\alpha}\hspace{0.1em}2^{\alpha}}{(1-\delta)h(x) + \delta(x(1-x))^{\alpha}}\\
& \stackrel{\mathclap{\mbox{\footnotesize(b)}}}{\le} 2^{-\rho_1}+\delta\frac{\hspace{0.1em}2^{\alpha}}{1-\delta} \frac{(x(1-x))^{\alpha}}{h(x)} \\
& \stackrel{\mathclap{\mbox{\footnotesize(c)}}}{\le} 2^{-\rho_1}
+ \sqrt{2} \frac{\delta}{1-\delta} c_3,
\end{split}
\end{equation}
where the inequality (a) uses that $t(x)\le 2^{\alpha}$ for any $x\in (0, 1)$; the inequality (b) uses that $h(x) \ge 0$ and $(x(1-x))^{\alpha} \ge 0$; and the inequality (c) uses that $\alpha \le 1/2$ and the definition of $c_3$ in \eqref{eq:defa}. By putting \eqref{eq:bLg1} and \eqref{eq:bLg2} together, we have that
\begin{equation}\label{eq:fin4}
L_g \le 2^{-\rho_1} + \sqrt{2} \frac{\delta}{1-\delta} c_3.
\end{equation} 
By combining \eqref{Z_n-L_g} and \eqref{eq:fin4}, the result for a general $\rm BMSC$ follows.  

Finally, consider the special case in which $W$ is a $\rm BEC$. Clearly, \eqref{eq:fin1} still holds, and, by using the definition \eqref{eq:recBEC} of the Bhattacharyya process $Z_n$ for the $\rm BEC$, in analogy to \eqref{eq:fin3}, we obtain that
\begin{equation}
{\mathbb E}[( Z_n(1-Z_n))^{\alpha}] \le \frac{1}{\delta}(L'_g)^n,
\end{equation}
where we define
\begin{equation*}
L'_g = \sup_{x\in (0, 1)} \frac{g(x^2)+g(2x-x^2)}{2g(x)}.
\end{equation*}
By using \eqref{eq:suphrhoBEC}, after some calculations, we have that 
\begin{equation*}
\frac{g_0(x^2)+g_0(2x-x^2)}{2g_0(x)} \le \frac{(1-\delta)h(x)2^{-\rho_1}+\delta (x(1-x))^{\alpha}\hspace{0.1em}t'(x)}{(1-\delta)h(x) + \delta(x(1-x))^{\alpha}},
\end{equation*}
with
\begin{equation*}
t'(x)=\displaystyle\frac{1}{2}\left( \bigl(x(1+x)\bigr)^\alpha + \bigl((2-x)(1-x)\bigr)^\alpha \right).
\end{equation*}
As $(1-x)\le (1-x)^{1/3}$ for any $x\in (0, 1)$, we obtain that $t'(x)\le t(x)$, with $t(x)$ defined in \eqref{eq:tx}. Therefore, the result for the $\rm BEC$ naturally follows.  
\end{proof}

\subsection{Sketch of the Proof of \eqref{eq:jointPeNZ}} \label{app:Battadep}

Eventually, let us briefly sketch how to prove the result stated in Remark \ref{rmk:Battadep}. The dependency on the Bhattacharyya parameter $Z(W)$ first appears in formula \eqref{eq:fin3}. Hence, under the hypothesis of Lemma \ref{lm:htoexp}, one can easily prove that
\begin{equation}\label{eq:newubZ}
{\mathbb E}\left[(Z_n(1-Z_n))^{\alpha}\right] \le \frac{g(Z(W))}{\delta}\left(2^{-\rho_1}+\sqrt{2}\frac{\delta}{1-\delta} c_3\right)^n,
\end{equation}
where $g(x)$ is defined in \eqref{eq:defgx}. Consequently, by following passages similar to those in the proof of Lemma \ref{lm:exptoBatta} in Appendix \ref{app:exptoBatta} and of Theorem \ref{th:scalingexp} in Section \ref{subsec:mathpart}, we conclude that  
\begin{equation}\label{eq:Zn0Batta1}
{\mathbb P}\left(Z_{n_0} \le Z(W)\cdot 2^{-2n_0}\right)\ge I(W)- c_8 \hspace{0.1em} 2^{-n_0/\mu},
\end{equation}
where $c_8$ is a constant. Note that, in formula \eqref{eq:lemma22}, $Z_{n_0+n_1}$ is upper bounded by a quantity that does not depend on $x$. In order to make this dependency appear, we use a procedure similar to that of the proof of Lemma 22 in \cite{HAU14}. As a result, we obtain that 
\begin{equation}\label{eq:Zn0Batta2}
{\mathbb P}\left(Z_{n_0+n_1} \le x^{\scriptstyle \frac{1}{2}\cdot 2^{\scriptstyle\sum_{i=1}^{n_1} B_i}}\mid Z_{n_0} = x\right)\ge 1- c_9 \hspace{0.1em} \sqrt{x}(1-\log_2 x),
\end{equation}
where $c_9$ is a constant. By combining \eqref{eq:Zn0Batta1} and \eqref{eq:Zn0Batta2}, the result follows by using arguments similar to those of the proof of Theorem \ref{th:unifexp} in Section \ref{subsec:proofjoint}.

\section*{Acknowledgment}
This work was supported by grant No. 200020\_146832/1 of the Swiss National Science Foundation. S. Hamed Hassani is supported by ERC Starting Grant under grant number 307036.

\bibliographystyle{IEEEtran}
\bibliography{lth,lthpub}

\end{document}